\documentclass[12pt]{article}
\usepackage{amsmath}
\usepackage{amssymb}
\usepackage{amsthm}
\usepackage{booktabs}
\usepackage{enumerate}
\usepackage{enumitem}
\usepackage{float}
\usepackage[singlelinecheck=false]{caption}

\usepackage{graphicx,psfrag,epsf}
\usepackage[space]{grffile}
\usepackage{natbib}
\usepackage{url} 

\newcommand{\blind}{1}

\newtheorem{lemma}{Lemma}

\newtheorem{proposition}{Proposition}

\DeclareMathOperator*{\argmin}{arg\,min}
\newcommand{\CV}{\operatorname{CV}}

\newcommand{\E}{\operatorname{E}}
\DeclareMathOperator*{\cov}{cov}
\DeclareMathOperator*{\tr}{tr}

\newcommand{\T}{T}

\newcommand{\OhP}{O_p}

\newcommand{\R}{\mathbb{R}}

\newcommand{\muX}{\mu^{X}}
\newcommand{\muY}{\mu^{Y}}
\newcommand{\bmuX}{\bar \mu^{X}}
\newcommand{\bmuY}{\bar \mu^{Y}}
\newcommand{\hmuY}{\hat \mu^{Y}}

\newcommand{\dataX}{\mathfrak{X}}

\newcommand{\Xtrain}{X_{\text{train}}}
\newcommand{\Ytrain}{Y_{\text{train}}}
\newcommand{\Xtest}{X_{\text{test}}}
\newcommand{\Ytest}{Y_{\text{test}}}

\newcommand{\hGX}{\hat G^{X}}
\newcommand{\hGY}{\hat G^{Y}}

\begin{document}

\def\spacingset#1{\renewcommand{\baselinestretch}%
{#1}\small\normalsize} \spacingset{1}


\if1\blind
{
  \title{\bf Estimating the number of clusters using cross-validation}
  \author{Wei Fu and Patrick O. Perry \\
  Stern School of Business, New York University}
  \maketitle
} \fi

\if0\blind
{
  \bigskip
  \bigskip
  \bigskip
  \begin{center}
    {\LARGE\bf Estimating the number of clusters using cross-validation}
\end{center}
  \medskip
} \fi

\bigskip
\begin{abstract}
Many clustering methods, including $k$-means, require the user to specify the
number of clusters as an input parameter. A variety of methods have been
devised to choose the number of clusters automatically, but they often rely on
strong modeling assumptions. This paper proposes a data-driven approach to
estimate the number of clusters based on a novel form of cross-validation. The
proposed method differs from ordinary cross-validation, because clustering is
fundamentally an unsupervised learning problem. Simulation and real data
analysis results show that the proposed method outperforms existing methods,
especially in high-dimensional settings with heterogeneous or heavy-tailed
noise. In a yeast cell cycle dataset, the proposed method finds a parsimonious
clustering with interpretable gene groupings.
\end{abstract}

\noindent%
{\it Keywords:} clustering, cross-validation, $k$-means, model selection, unsupervised learning
\vfill

\newpage
\spacingset{1.45} 
\section{Introduction}
\label{sec:intro}

A clustering procedure segments a collection of items into smaller groups,
with the property that items in the same group are more similar to each other
than items in different groups \citep{hartigan1975clustering}.  Such
procedures are used in two main applicaitons: (a) exploratory analysis, where clusters
reveal homogeneous sub-groups within a large sample; (b) data
reduction, where high-dimensional item attribute vectors get reduced to discrete
cluster labels \citep{jain1999data}.

With many clustering methods, including the popular $k$-means clustering
procedure, the user must specify $k$, the number of clusters
\citep{jain2010data}.  One popular ad-hoc device for selecting the number of
clusters is to use an analogue of the principal components scree plot: plot
the within-cluster dispersion $W_k$, as a function of the number of clusters
$k$, looking for an ``elbow'' in the plot.  This approach is simple and often
performs well, but it requires subjective judgment as to where the elbow is
located, and as we demonstrate in Appendix~\ref{sec:elbow-fail}, the approach
can easily fail.  In this report, we propose a new method to choose the number
of clusters automatically.

The problem of choosing $k$ has been well-studied, and dozens of methods have
been proposed~\citep{chiang2010intelligent,fujita2014non}. The main difficulty
in choosing $k$ is that clustering is fundamentally an ``unsupervised''
learning problem, meaning that there is no obvious way to use ``prediction
ability'' to drive the model selection \citep{hastie2009elements}.  Most
existing methods for choosing $k$ instead rely on explicit or implicit
assumptions about the data distribution, including it shape, scale, and
correlation structure.

Several authors advocate choosing $k$ by performing a sequence of hypothesis
tests with null and alternative hypotheses of the form $H_0 : k = k_0$ and
$H_1: k > k_0$, starting with $k_0 = 1$ and proceeding sequentially with
higher values of $k_0$ until a test fails to reject $H_0$. The gap statistic
method typifies this class of methods, with a test statistic that measures the
within-cluster dispersion relative to what is expected under a reference
distribution \citep{tibshirani2001estimating}.

Other authors have proposed choosing $k$ by using information
criteria.  For example, \citet{sugar2003finding} proposed an approach that
minimizes the estimated ``distortion'', the average distance per dimension.
Likewise, \citet{fraley2002model}'s model-based method fits Gaussian mixture model
models to the data, then selects the number of mixture components, $k$, using
the Bayesian Information Criterion (BIC).

A third set of approaches is based on the idea of ``stability'', that clusters
are meaningful if they manifest in multiple independent samples from the same
population. \citet{ben2001stability}, \citet{tibshirani2005cluster},
\citet{wang2010consistent} and \citet{fang2012selection} developed methods
based on this idea.

The procedure we propose in this report is based on a form of
cross-validation, and it is adaptive to the characteristics of the data
distribution. The essential idea is to devise a way to measure a form of
internal prediction error associated each choose of $k$, and then choose the
$k$ with the smallest associated error. We describe this method in detail in
Section~\ref{sec:meth}.  In Section~\ref{sec:self-consistent}, we prove that
our method is self-consistent. Then, in Section~\ref{sec:theory-gaussian}, we
analyze the performance of our method in the presence of Gaussian noise. The
theoretical analysis shows that the performance of our method degrades in the
presence of correlated noise; to fix this, we propose a correction for
correlated noise in Section~\ref{sec:corr-correct}. In
Sections~\ref{sec:simulations} and~\ref{sec:empirical-validation}, we
demonstrate that our method is competitive with other state-of-the-art
procedures in both simulated and real data sets. Then, in
Section~\ref{sec:application-yeast}, we apply our method to a Yeast cell cycle
dataset. We conclude with a short discussion in Section~\ref{sec:discussion}.

\section{Cross-validation for clustering}
\label{sec:meth}

\subsection{Problem statement}

Suppose that we are given a data matrix with $N$ rows and $P$ columns, and we
are tasked with choosing an appropriate number $k$ of clusters to use for
performing $k$-means clustering on the rows of the data matrix. Recall that
the $k$-means procedure takes a set of observations $\{ x_1, \dotsc ,x_n \}$
and finds a set of $k$ or cluster centers $A = \{ a_1, \dotsc, a_k \} $
minimizing the within cluster dispersion
\[
  W(A) = \sum_{i=1}^{n} \min_{a \in A} \|x_i - a\|^2.
\]
This implicitly defines a cluster assignment rule
\[
  g(x) = \argmin_{g \in \{1, \dotsc, k\}} \|x - a_g\|^2,
\]
with ties broken arbitrarily.

We can consider the problem of choosing $k$, the number of clusters, to be a
model selection problem. In other domains, especially supervised learning
problems like regression and classification, cross-validation is popular for
performing model selection.
In these settings, the data comes in the form of $N$
predictor-response pairs, $(X_1, Y_1), \dotsc, (X_N, Y_N)$, with $X_i \in
\R^{p}$ and $Y_i \in \R^{q}$.  The data can be represented as a matrix with
$N$ rows and $p + q$ columns.  We partition the data into $K$ hold-out
``test'' subsets, with $K$ typically chosen to be $5$ or $10$.  For each
``fold'' $r$ in the range $1, \dotsc, K$, we permute the rows of the data
matrix to get $\dataX$, a matrix with the $r$th test subset as its trailing
rows.  We partition $\dataX$ as
\[
  \dataX =
  \begin{bmatrix}
    \Xtrain & \Ytrain \\
    \Xtest  & \Ytest
  \end{bmatrix}.
\]
We use the training rows $[ \Xtrain\ \Ytrain ]$ to fit a regression model
$\hat Y = \hat Y(X)$, and then evaluate the performance of this model on the
test set, computing the cross-validation error $\|\Ytest - \hat Y(\Xtest)\|^2$
or some variant thereof.  We choose the model with the smallest
cross-validation error, averaged over all $K$ folds.

In unsupervised learning problems like factor analysis and clustering, the
features of the observations are not naturally partitioned into ``predictors''
and ``responses'', so we cannot directly apply the cross-validation procedure
described above.  For factor analysis, there are at least two versions of
cross-validation.  \citet{wold78cross} proposed a ``speckled'' holdout, where
in each fold we leave out a subset of the elements of the data matrix.  Wold's
procedure works well empirically, but does not have any theoretical support,
and it requires a factor analysis procedure that can handle missing data.
\citet{owen2009bi} proposed a scheme called ``bi-cross-validation'' wherein
each fold designates a subset of the data matrix columns to be response and a
subset of the rows to be test data.  This generalized a procedure due to
\citet{gabriel2002biblot}, who proposed holding out a single column and a
single row at each fold.


In the sequel, we extend Gabriel cross-validation to the problem of selecting
the number of clusters, $k$, automatically, and we provide theoretical and
empirical support analogous to the consistency results proved by
\citet{owen2009bi}.

\subsection{Gabriel cross-validation}
\label{sec:gabriel-cv-algorithm}

Our version of Gabriel cross validation for clustering works by performing a
sequence of ``folds'' over the data. We use these folds to estimate a version
of prediction error (cross-validation error) for each possible value of $k$;
we then choose the value $\hat k$ with the smallest cross-validation error.

In each fold of our cross-validation procedure, we permute the rows and
columns of the data matrix and then partition the rows and columns as $N = n +
m$ and $P = p + q$ for positive integers $n$, $m$, $p$, and $q$.  We treat
the first $p$ columns as ``predictors'' and the last $q$ columns as
``responses''; similarly, we treat the first $n$ rows as ``train''
observations and the last $m$ rows as ``test'' observations.  In block form,
the permuted data matrix is
\[
  \dataX
  =
  \begin{bmatrix}
    \Xtrain & \Ytrain \\
    \Xtest  & \Ytest
  \end{bmatrix},
\]
where
$\Xtrain \in \R^{n \times p}$,
$\Ytrain \in \R^{n \times q}$,
$\Xtest \in  \R^{m \times p}$,
and
$\Ytest \in  \R^{m \times q}$.

Given such a partition of $\dataX$, we perform four steps for each value of
$k$, the number of clusters:
\begin{enumerate}
  \item \label{step:gabriel-cluster}
    \textbf{Cluster:}
    Cluster $Y_{1}, \dotsc, Y_n$, the rows of $\Ytrain$, yielding the
    assignment rule $\hGY : \R^q \to \{ 1, \dotsc, k \}$ and the
    cluster means $\bmuY_1, \dotsc, \bmuY_k$.  Set $\hGY_i = \hGY(Y_i)$ to
    be the assigned cluster for row $i$.
  \item \label{step:gabriel-classify}
    \textbf{Classify:}
    Take $X_{1}, \dotsc, X_n$, the rows of $\Xtrain$ to be predictors,
    and take $\hGY_1, \dotsc, \hGY_n$ to be corresponding class labels.  Use
    the pairs $\{ (X_i, \hGY_i) \}_{i=1}^{n}$ to train a classifier
    $\hGX : \R^p \to \{ 1, \dotsc, k \}$.
  \item \label{step:gabriel-predict}
    \textbf{Predict:}
    Apply the classifier to $X_{n+1}, \dotsc, X_{n+m}$, the rows of
    $\Xtest$, yielding predicted classes $\hGX_i = \hGX(X_i)$ for
    $i = n+1, \dotsc, n+m$.  For each value of $i$ in this range, compute
    predicted response $\hat Y_i = \bmuY(\hGX_i)$, where
    $\bmuY(g) = \bmuY_g$.
  \item \label{step:gabriel-evaluate}
    \textbf{Evaluate:}
    Compute the cross-validation error
    \[
      \CV(k) = \frac{1}{m} \sum_{i=n+1}^{n+m} \|Y_i - \hat Y_i\|^2,
    \]
    where $Y_{n+1}, \dotsc, Y_{n+m}$ are the rows of $\Ytest$.
\end{enumerate}
\noindent
In principle, we could use any clustering and classification methods in
steps~\ref{step:gabriel-cluster} and~\ref{step:gabriel-classify}.  In this
report, we use $k$-means \citep{hartigan1979algorithm} as the clustering algorithm 
and develop the theoretical properties of the proposed method based on $k$-means. 
For the classification step, we compute the mean value of $X$ for each class; we assign an
observation to class $g$ if that class has the closest mean (randomly breaking
ties between classes).  The classification step is equivalent to linear
discriminant analysis with equal class priors and identity noise covariance
matrix.

To choose the folds, we randomly partition the rows and columns into $K$ and
$L$ subsets, respectively.  Each fold is indexed by a pair $(r,s)$ of
integers, with $r \in \{1, \dotsc, K\}$ and $s \in \{1, \dotsc, L\}$.  Fold
$(r,s)$ treats the $r$th row subset as ``test'', and the $s$th column subset
as ``response''.  We typically take $K = 5$ and $L = 2$.  For the number of
clusters, we select the value of $k$ that minimizes the average of $\CV(k)$
over all $K \times L$ folds (choosing the smallest value of $k$ in the event
of a tie).

In Section~\ref{sec:self-consistent}, we prove that this procedure is
self-consistent, in the sense that it recover the correct value of $k$ in the
absence of noise. Then, in Section~\ref{sec:theory-gaussian}, we analyze some
of the properties of Gabriel cross-validation in the presence of Gaussian
noise.

\section{Self-consistency}
\label{sec:self-consistent}

An important property of any estimation procedure is that in the absence of of
noise, the procedure correctly estimates the truth. This property is called
``self-consistency'' \citep{tarpey96}. We will now show that Gabriel
cross-validation is self-consistent. That is, in the absence of noise, the
Gabriel cross-validation procedure finds the optimal number of clusters.

It will suffice to prove self-consistency for a single fold of the
cross-validation procedure.  As in section~\ref{sec:gabriel-cv-algorithm} we
assume that the $P$ variables of the data set have been partitioned into $p$
predictor variables represented in vector~$X$ and $q$ response variables
represented in vector~$Y$.  The $N$ observations have been divided into two
sets: $n$ train observations and $m$ test observations.  We state the
assumptions for the self-consistency result in terms of a specific split; for
the result to hold in general, with high probability, these assumptions would
have to hold with high probability for a random split. The following theorem
gives conditions for Gabriel cross-validation to recover the true number of
clusters in the absence of noise.

\begin{proposition}\label{prop:self-consistency}

Let $\{ (X_i, Y_i) \}_{i=1}^{n+m}$ be the data from a single fold of Gabriel
cross-validation.  For any $k$, let $\CV(k)$ be the cross-validation error for
this fold, computed as described in Section~\ref{sec:gabriel-cv-algorithm}.
We will assume that there are $K$ true centers $\mu(1), \dotsc,\mu(K)$, with
the $g$th cluster center partitioned as $\mu(g) = \bigl(\muX(g),
\muY(g)\bigr)$ for $g = 1, \dotsc, K$.  Suppose that
\begin{enumerate}[label=(\roman*)]
  \item \label{asn:self-consistency-noiseless}
    Each observation $i$ has a true cluster $G_i \in \{ 1, \dotsc, K \}$.
    There is no noise, so that $X_i = \muX({G_i})$ and $Y_i = \muY(G_i)$ for
    $i = 1, \dotsc, n+m$.
  \item \label{asn:self-consistency-distinct-mux}
    The vectors $\muX(1), \dotsc,\muX(K)$ are all distinct.
  \item \label{asn:self-consistency-distinct-muy}
    The vectors $\muY(1), \dotsc,\muY(K)$ are all distinct.
  \item \label{asn:self-consistency-train}
    The training set contains at least one member of each cluster: for all $g$
    in the range $1, \dotsc, K$, there exists at least one $i$ in the range
    $1, \dotsc, n$ such that $G_i = g$.
  \item \label{asn:self-consistency-test}
    The test set contains at least one member of each cluster: for all $g$ in
    the range $1, \dotsc, K$, there exists at least one $i$ in the range $n+1,
    \dotsc, n+m$ such that $G_i = g$.
\end{enumerate}
Then $\CV(k) < \CV(K)$ for $k < K$, and $\CV(k) = \CV(K)$ for $k > K$, so that
Gabriel cross-validation correctly chooses $k = K$.
\end{proposition}

The proposition states that our method works well in the absence of noise,
when each observation is equal to its cluster center.  The essential
assumption here is assumption~\ref{asn:self-consistency-noiseless}, which states
that there is no noise.  If we are willing to assume, say, that the cluster
centers~$\mu(g) = \bigl(\muX(g),\muY(g)\bigr)$ for $g = 1, \dotsc, K$ were
randomly drawn from a distribution with a density over $\R^{p+q}$, then
assumptions~\ref{asn:self-consistency-distinct-mux}
and~\ref{asn:self-consistency-distinct-muy} will hold with probability one for
all splits of the data. Likewise, if the clusters are not too small (relative
to $n$ and $m$), then assumptions~\ref{asn:self-consistency-train}
and~\ref{asn:self-consistency-test} will likely hold for a random split of the
data into test and train.

Proposition~\ref{prop:self-consistency} follows from
Lemmas~\ref{lem:self-consistency1} and~\ref{lem:self-consistency2}, which we
now state and prove.

\begin{lemma}\label{lem:self-consistency1}
Suppose that the assumptions of Proposition~\ref{prop:self-consistency} are in
force.  If $k < K$, then $\CV(k) > 0$.
\end{lemma}
\begin{proof}
By definition,
\[
  \CV(k)
    =
      \sum_{i=n+1}^{n+m}
        \| Y_i - \bmuY (\hGX_i) \|^2,
\]
where $\bmuY(g)$ is the center of cluster $g$ returned from applying $k$-means
to $Y_1, \dotsc, Y_n$.  Assumptions~\ref{asn:self-consistency-noiseless}
and~\ref{asn:self-consistency-test}, imply that as $i$ ranges over the test
set $n+1, \dotsc, n+m$, the response $Y_i$ ranges over all distinct values in
$\{ \muY(1), \dotsc, \muY(K) \}$.
Assumption~\ref{asn:self-consistency-distinct-muy} implies that there are
exactly $K$ such distinct values.  However, there are only $k$ distinct values
of $\bmuY(g)$.  Thus, at least one summand
\(
  \| Y_i - \bmuY(\hGX_i) \|^2
\)
is nonzero.  Therefore,
\(
  \CV(k) > 0.
\)
\end{proof}

\begin{lemma}\label{lem:self-consistency2}
Suppose that the assumptions of Proposition~\ref{prop:self-consistency} are in
force.  If $k \geq K$, then $\CV(k) = 0$.
\end{lemma}
\begin{proof}
From assumptions~\ref{asn:self-consistency-noiseless},
\ref{asn:self-consistency-distinct-muy},
and~\ref{asn:self-consistency-train}, we know the cluster centers
gotten from applying $k$-means to $Y_1, \dotsc, Y_n$ must include
$\muY(1), \dotsc, \muY(K)$.  Without loss of generality, suppose that
$\bmuY(g) = \muY(g)$ for $g = 1, \dotsc, K$.  This implies that
$\hGY_i = G_i$ for $i = 1, \dotsc, n$.  Thus, employing
assumption~\ref{asn:self-consistency-noiseless} again, we get that
$\bmuX(g) = \muX(g)$ for $g = 1, \dotsc, K$.

Since assumption~\ref{asn:self-consistency-distinct-mux} ensures that
$\muX(1), \dotsc, \muX(K)$ are all distinct, we must have that $\hGX_i = G_i$
for all $i = 1, \dotsc, m+n$.  In particular, this implies that $\bmuY(\hGX_i)
= Y_i$ for $i = 1, \dotsc, m+n$, so that $\CV(k) = 0$.
\end{proof}

\section{Analysis under Gaussian noise}
\label{sec:theory-gaussian}

\subsection{Single cluster in two dimensions}

Proposition~\ref{prop:self-consistency} tells us that the Gabriel cross-validation
method recovers the true number of clusters when the noise is negligible.
While this result gives us some assurance that the procedure is well-behaved,
we can bolster our confidence and gain insight into its workings by analyzing
its behavior in the presence of noise.  We first study the case of a single
cluster in two dimensions with correlated Gaussian noise.

\begin{proposition}\label{prop:single-2d}
Suppose that $\{ (X_i, Y_i) \}_{i=1}^{n + m}$ is data from a single fold
of Gabriel cross-validation, where each $(X,Y)$ pair in $\R^2$ is an
independent draw from a mean-zero multivariate normal distribution with unit
marginal variances and correlation $\rho$.  In this case, the data are drawn
from a single cluster; the true number of clusters is~$1$.  If $|\rho| < 1/2$,
and $k > 1$, then $\CV(1) < \CV(k)$ with probability tending to one as $m$
and $n$ increase.
\end{proposition}

\begin{proof}

Throughout the proof we will assume that $\rho \geq 0$; a similar argument
holds with minor modification when $\rho < 0$.

Set $\hGY_1, \dotsc, \hGY_n$ to be the cluster labels gotten from applying
$k$-means to $Y_1, \dotsc, Y_n$. Denote the cluster means by $\bmuY_1 \leq
\bmuY_2 \leq \dotsb \leq \bmuY_k$.  Pollard's \citeyearpar{pollard1981strong}
strong consistency theorem for $k$-means implies that for large $n$, the
cluster centers are close to population clusters centers $a_1 < a_2 < \dotsb <
a_k$. Specifically, $\bmuY_j = a_j + \OhP(n^{-1/2})$.  Since the distribution
of $Y$ is symmetric, the population centers $a_1, a_2, \dotsc, a_k$ are
symmetric about the origin.

For $j$ in $1, \dotsc, k$, set
\[
  \bmuX_j = \frac{\sum_{i=1}^{n} 1\{ \hGY_i = j \} X_i}
                 {\sum_{i=1}^{n} 1\{ \hGY_i = j \}}.
\]
The classification rule $\hGX$ is defined by
\(
  \hGX(X) = \argmin_j \| \bmuX_j - X \|.
\)
Denote the boundaries between the population clusters as
$b_j = (a_j + a_{j+1})/2$ for $j = 1, \dotsc, k-1$.
Set $b_0 = -\infty$ and $b_k = +\infty$. 
Then, $\bmuX_j$ is within $\OhP(n^{-1/2})$ of the following expectation:
\begin{align*}
  \E(X \mid b_j \leq Y \leq b_{j+1})
    &= \E\{ E(X \mid Y) \mid b_j \leq Y \leq b_{j+1}\} \\
    &= \E\{ \rho Y \mid b_j \leq Y \leq b_{j+1}\} \\
    &= \rho \E\{ Y \mid b_j \leq Y \leq b_{j+1}\} \\
    &= \rho a_j.
\end{align*}
That is, $\bmuX_j = \rho a_j + \OhP(n^{-1/2})$. For $j = 1, \dotsc, k-1$,
the boundary between sample the classification based on $X$ to
labels $j$ and $j+1$ is
$(\bmuX_j + \bmuX_{j+1}) /2 = \rho b_j + \OhP(n^{-1/2})$.

Set $\hGX_i = \hGX(X_i)$.
The cross-validation error is
\begin{align*}
  \CV(k)
  &= \frac{1}{m} \sum_{i=n+1}^{n+m} \|Y_i - \hat Y_i\|^2
  \\
  &= \frac{1}{m} \sum_{i=n+1}^{n+m}
  \sum_{j=1}^{k} \|Y_i - \bmuY_j\|^2 1\{ \hGX_i = j \}
  \\
  &=
  \sum_{j=1}^{k}
    \pi_j
    \E[ \|Y - a_j\|^2 \mid \rho b_j < X < \rho b_{j+1}]
  + \OhP(n^{-1/2})
  + \OhP(m^{-1/2}),
\end{align*}
where $\pi_j = \Pr(\rho b_j < X < \rho b_{j+1}).$

For $j = 1, \dotsc, k$,
set
\(
  \hmuY_j =
  \E[ Y  \mid \rho b_j < X < \rho b_{j+1}].
\)
Note that
\[
  \CV(1)
    = 
  \sum_{j=1}^{k}
    \pi_j
    \E[ \|Y - 0\|^2 \mid \rho b_j < X < \rho b_{j+1}]
  + \OhP(n^{-1/2})
  + \OhP(m^{-1/2}).
\]
Thus, the difference in cross-validation errors is
\begin{align*}
\CV(k) - \CV(1)
  &=
  \sum_{j=1}^{k}
    \pi_j
    a_j
    (a_j - 2 \hmuY_j)
  + \OhP(n^{-1/2})
  + \OhP(m^{-1/2}).
\end{align*}
For arbitrary $j$,
\begin{align*}
  \hmuY_j &=
  \E[ E(Y\mid X)  \mid \rho b_j < X < \rho b_{j+1}]
\\
  &= \rho \E[ X  \mid \rho b_j < X < \rho b_{j+1}].
\end{align*}
Since $0 \leq \rho \leq 1$,
in cases where $0 \leq b_j < b_{j+1}$, we have that
$\hmuY_j \leq \rho a_j$; similarly, when $b_j < b_{j+1} \leq 0$, it follows
that $\hmuY_j \geq \rho a_j$. In either of these two situations, if $\rho <
1/2$, then
\[
  a_j (a_j - 2 \hmuY_j) > 0.
\]
The last situation to consider is when $b_j < 0 < b_{j+1}$, in which case
$b_j = - b_{j+1}$; here, $\hmuY_j = a_j = 0$.
Putting this all together, we have that as $n$ and $m$ tend to infinity,
the probability that
\(
  \CV(k) > \CV(1)
\)
tends to one.
\end{proof}

We confirm the result of Proposition~\ref{prop:single-2d} with a simulation.
We perform $10$ replicates.  In each replicate, we generate $20000$
observations from a mean-zero bivariate normal distribution with unit marginal
variances and correlation $\rho$.  We perform a single $2 \times 2$ fold of
Gabriel cross-validation and report the cross-validation mean squared error
for the number of clusters $k$ ranging from $1$ to $5$.
Figure~\ref{fig:nullcorr-equal} shows the cross-validation errors for all $10$
replicates.  The simulation demonstrates that in the Gabriel cross-validation
criterion chooses the correct answer $k = 1$ whenever $\rho < 0.5$; the
criterion chooses $k \geq 2$ clusters whenever $|\rho| > 0.5$.

\begin{figure}
\centering
\includegraphics[width=3in]{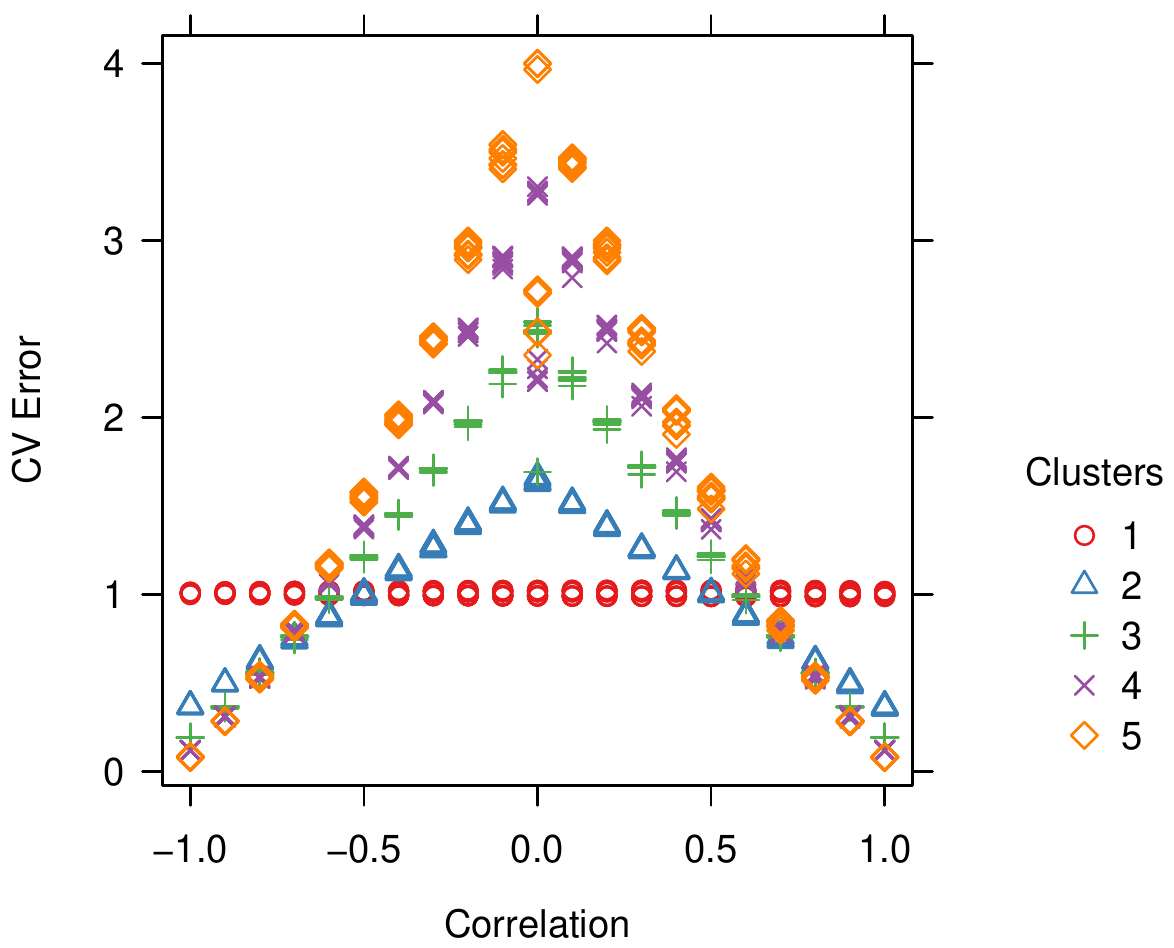}
\caption{Cross-validation error on $10$ replicates, with the number of
clusters $k$ ranging from $1$ to $5$.  Data is generated from two-dimensional
multivariate normal distribution with correlation $\rho$.  The Gabriel
cross-validation criterion chooses the correct answer $k = 1$ whenever
$|\rho| < 0.5$; the criterion chooses $k \geq 2$ clusters whenever $|\rho| > 0.5$.}
\label{fig:nullcorr-equal}
\end{figure}

Intuitively, when the correlation is high, the response feature, $Y$, looks
similar to the predictor feature, $X$. Prediction error on $X$ always
decreases with larger $k$. Thus, when the correlation is high, the prediction
error for $Y$ will also decrease with larger $k$. This explains why
cross-validation breaks down in the presence of strong correlation.

In Appendix~\ref{sec:single-general-dim}, using similar techniques to those used
to prove Proposition~\ref{prop:single-2d}, we derive an analogous result for correlated
Gaussian noise in more than two dimensions. A similar phenomenon holds:
the Gabriel cross-validation method fails when the first principal component
of the $Y$ variables is strongly correlated with a linear combination of the
$X$ variables.

Proposition~\ref{prop:single-2d} tells us that Gabriel cross-validation fails
when there is strong correlation between the variables. To get around this, in
practice we will transform the data to reduce correlation before performing
cross-validation. We detail this approach in Section~\ref{sec:corr-correct}.

\subsection{Two clusters in two dimensions}

We will now analyze a simple two-cluster setting, and derive conditions for
Gabriel cross-validation to correctly prefer $k=2$ clusters to $k=1$.  The
main assumption of the proposition is that the cluster centers are not too
close. The precise definition of ``too close'' is stated in terms of
$\Phi(\cdot)$ and $\varphi(\cdot)$, the standard normal cumulative
distribution function and density, respectively.  The inequality is hard to
interpret directly, but we show the boundary between ``too close'' and ``well
separated'' in Fig.~\ref{fig:overlap-color_plot}, after the proof of the
proposition.

\begin{proposition}\label{prop:twoclust}
Suppose that $\{(X_i,Y_i)\}_{i=1}^{n+m}$ is data from a single fold of Gabriel
cross-validation, where each $(X,Y)$ pair in $\R^2$ is an independent draw
from an equiprobable mixture of two multivariate normal distributions with
identity covariance. Suppose that the first mixture component has mean
$\mu = (\muX, \muY)$ and the second has mean $-\mu = (-\muX, -\muY)$,
where $\muX \geq 0$ and $\muY \geq 0$. 
If the cluster centers are well separated, specifically such that
\(
  2 \varphi(\muY) + \muY + 2 \muY \Phi(\muY) < 4 \muY \Phi(\muX),
\)
then $\CV(2) < \CV(1)$ with probability tending to one as $m$ and $n$ increase.
\end{proposition}

\begin{proof}
There are two clusters: observations from cluster~$1$ are distributed as
\(
  \mathcal{N}(\mu, I)
\)
and observations from cluster~$2$ are distributed as
\(
\mathcal{N}(-\mu, I)
\)
where $\mu = (\muX, \muY)$.  Without loss of generality, $\muX \geq 0$ and
$\muY \geq 0$.  Let $G_i$ be the true cluster of observation~$i$
where, by assumption,
\[
  \Pr(G_i=1) = \Pr(G_i=2) = 1/2.
\]
After applying $k$-means to $\{ Y_i\}_{i=1}^{n}$ with $k=2$, if $n$ is large
enough, then the estimated cluster means $\bar{\mu}^Y_1$ and $\bar{\mu}^Y_2$
will be close to $\E(Y \mid Y>0)$ and $\E(Y \mid Y < 0)$, with errors of
size $\OhP(n^{-1/2})$. To compute these quantities, let
$(X_1, Y_1) \sim \mathcal{N}(\mu, I)$ and $(X_2, Y_2) \sim \mathcal{N}(-\mu, I)$
be draws from the mixture components, and let $(X,Y)$ be defined such that
$\Pr(X= X_1, Y = Y_1) = \Pr(X= X_2, Y = Y_2) = 1/2$. Then,
\begin{align*}
  \E(Y \mid Y>0)
  &= \E(Y_1 \mid Y_1 > 0) \cdot \Pr(Y = Y_1 \mid Y >0)
   + \E(Y_2 \mid Y_2 > 0) \cdot \Pr(Y = Y_2 \mid Y > 0)
\\
&=  \{ \muY + \varphi(\muY) / \Phi(\muY) \} \cdot \Phi(\muY)
+ [-\muY + \varphi(\muY) / \{ 1 - \Phi(\muY) \}] \cdot \{ 1 - \Phi(\muY) \}
\\
  &= 2\varphi(\muY)+ 2\muY\Phi(\muY)- \muY. 
\end{align*}
In the second line, we have used
Lemma~\ref{lem:truncated-normal-moments} from
Appendix~\ref{app:technical-lemmas} to compute the conditional expectations;
$\varphi()$ and $\Phi()$ are the standard normal density and
cumulative distribution function, respectively. By symmetry,
\begin{align*}
  \E(Y \mid Y<0) &= -\E(Y \mid Y > 0).
\end{align*}

The classification rule learned from the training data
$\{(X_i, \hGY_i)\}_{i=1}^{n}$ will have its decision boundary at $0 + \OhP(n^{-1/2})$;
that is, in the limit,
observations will get classified as coming from cluster~1 when $X > 0$. 
Set $a = \E(Y \mid Y > 0)$.
Up to terms of order $\OhP(n^{-1/2})$,
the cross-validation error from a single observation is distributed as
\[
  (Y - a)^2 1\{ X > 0 \} + (Y + a)^2 1\{X < 0\}.
\]
Using the fact that conditional on the mixture component, the $X$ and $Y$
coordinates are independent, we can compute the expectation of the first
summand as
\begin{align*}
  \E[(Y - a)^2 1\{ X > 0 \}]
  &= (1/2) \E[(Y_1 - a)^2] \Pr(X_1 > 0)
  + (1/2) \E[(Y_2 - a)^2] \Pr(X_2 > 0)
\\
&= (1/2) \Phi(\muX) \E(Y_1 - a)^2  + (1/2) \{ 1 - \Phi(\muX) \} \E(Y_2 - a)^2
\\
&=
(1/2) [ 1  + \Phi(\muX) (\muY - a)^2 + \{ 1 - \Phi(\muX) \} (-\muY - a)^2 ].
\end{align*}
By a similar calculation, the expectation of the second summand is
\[
  \E[(Y + a)^2 1\{ X < 0 \}]
  =
  (1/2)[ 1 + \{ 1 - \Phi(\muX) \} (\muY + a)^2 + \Phi(\muX) (-\muY + a)^2 ].
\]
Adding the two terms, we get that the expected cross-validation error from a
single observation is
\[
  1 + \Phi(\muX) (\muY - a)^2 + \{ 1 - \Phi(\muX) \} (\muY + a)^2
    =
  1 + (\muY)^2 + a \, \{  a + 2 \muY - 4 \muY \Phi(\muX) \}.
\]
Thus, the $k = 2$ cross-validation error on the test set is
\begin{align*}
  \CV(2)
  &=
    \frac{1}{m}
    \sum_{i=n+1}^{n+m}
      \| (Y_i - \bmuY_1) 1\{\hGX_i = 1\} \|^2
      +
      \| (Y_i - \bmuY_2) 1\{\hGX_i = 2\} \|^2
\\
&=
1 + (\muY)^2 +  a \, \{ a + 2 \muY - 4 \muY \Phi(\muX) \}
+ \OhP(n^{-1/2}) + \OhP(m^{-1/2}).
\end{align*}

When $k=1$, the $k$-means centroid is equal to the sample mean $\bar Y_n = (1/n)
\sum_{i=1}^{n} Y_i$, approximately equal to 
$\E(Y) = 0$, with error of size $\OhP(n^{-1/2})$. The cross-validation error is
\[
  \CV(1) = \frac{1}{m} \sum_{i=n+1}^{n+m} \| Y_i - \bar Y_n \|^2
         = 1 + (\mu^Y)^2 + \OhP(m^{-1/2}) + \OhP(n^{-1/2}).
\]
Thus, if $a + 2 \muY - 4 \muY \Phi(\muX) < 0$, then $\CV(2) < \CV(1)$ with
probability tending to one as $m$ and $n$ increase. Substituting the
expression for $\E(Y \mid Y > 0)$ in place of $a$, the inequality holds
precisely when
\(
  2 \varphi(\muY) + \muY + 2 \muY \Phi(\muY) < 4 \muY \Phi(\muX).
\)
\end{proof}
We confirm the result of Proposition~\ref{prop:twoclust} with a simulation.
We perform $10$ replicates for each $(\mu^X, \mu^Y)$ pair, sweeping over
a two-dimensional grid of values in the domain $[0,3] \times [0,3]$, with step
size $0.1$ in each dimension.  In each
replicate, we generate $N=20000$ observations from an equiprobably mixture
of multivariate normal distributions with identity covariance, with one
component having mean $(\muX, \muY)$ and
the other component having mean $(-\muX, -\muY)$. We perform a single $2 \times 2$
fold of Gabriel cross-validation and report the number of times
(out of $10$ replicates) where $k=2$ is selected by the algorithm instead of $k=1$.
Figure~\ref{fig:overlap-color_plot} shows the frequency with which $k=2$ is selected by
the algorithm for each $(\mu^X, \mu^Y)$ pair. Darker (red) colors indicate
higher numbers (close to $10$), situations where $k = 2$ is selected more
often than $k = 1$. Ligher (blue) colors indicate that $k = 1$ is preferred.
We can see the simulation result perfectly align with the
theoretical curve (the black line), which separates the $k=2$ zone from the
$k=1$ zone.

\begin{figure}
\centering
\includegraphics[width=3in]{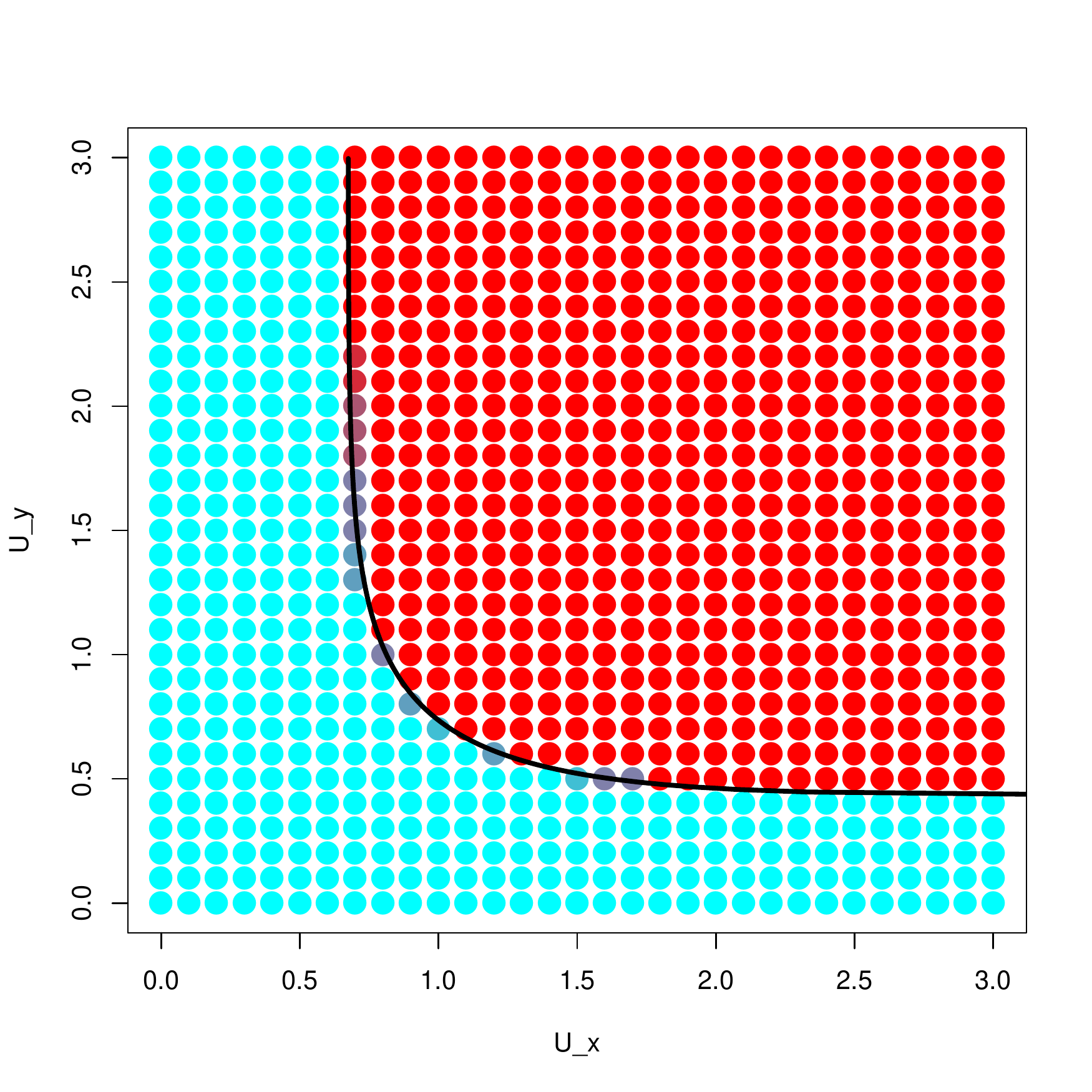}
\caption{Number of times $k=2$ is selected out of $10$ replicates for each
pair of $(\muX, \muY)$. The heat map shows the frequency $k=2$ is selected by
the algorithm, with light (blue) means $k = 1$ is preferred to $k = 2$, and
dark (red) indicates $k = 2$ is preferred to $k = 1$.
The black line is the theoretical boundary determined from
Proposition~\ref{prop:twoclust}.}
\label{fig:overlap-color_plot}
\end{figure}

\section{Adjusting for correlation}
\label{sec:corr-correct}

Proposition~\ref{prop:single-2d} shows that when the correlation between
dimensions is high, the Gabriel cross-validation method tends to overestimate
the number of clusters,~$k$. To mitigate this effect, we propose a two-stage
estimation procedure that attempts to transform the data to minimize the
correlation between features. In the first stage, we get a preliminary
estimate for the number of clusters,~$\hat k_0$, and we use this value to get
an estimate of the noise covariance matrix. Then, in the second stage, we
transform the data attempting to sphere the noise covariance, and re-estimate
the number of clusters, getting a final estimate~$\hat k$.

The details of the correlation correction procedure are as follow:

\begin{enumerate}

  \item Apply the Gabriel cross-validation method on the original data
    $\dataX$ to get a preliminary estimate of the number of
    clusters,~$\hat k_0$.

  \item Apply $k$-means to the full data set with observations
    $\dataX_1, \dataX_2, \dotsc, \dataX_N$ using $k = \hat k_0$ clusters.
    For $i = 1, \dotsc, N$, let $\hat \mu_i$ denote the assigned cluster
    mean for the $i$th observation.

	\item Estimate the noise covariance matrix $\hat{\Sigma}$:
    \[
      \hat{\Sigma}
        =
        \frac{1}{N - \hat k_0}
        \sum^N_{i=1}
        (\dataX_i-\hat{\mu}_i)(\dataX_i-\hat{\mu}_i)^\T.
    \]

  \item \label{step:transform}
    Compute the eigendecomposition $\hat{\Sigma} = \Gamma\Lambda\Gamma^\T$.
    Choose a random (Haar distributed) $P \times P$ orthogonal matrix $Q$.
    Rescale and rotate the original data matrix $\dataX$ to get a
    transformed data matrix defined by
    \[
      \tilde{\dataX} = \dataX\Gamma\Lambda^{-1/2}Q.
    \]

  \item Apply Gabriel cross-validation method to transformed data matrix
    $\tilde{\dataX}$ to get a final estimate for the number of
    clusters,~$\hat k$. 

\end{enumerate}

The noise covariance estimate assumes a shared covariance matrix for all
$k$~clusters. Letting $G_i$ denote the cluster membership of the $i$th
observation, and letting $\mu(g)$ denote the mean of cluster $g$ for $g = 1,
\dotsc, k$, the model supposes that
\[
  \dataX_i = \mu(G_i)+\varepsilon_i,
\]
where~$\varepsilon_i$ has mean zero and covariance matrix $\Sigma$, independent
of~$G_i$. If we knew $\Sigma$, then we could transform the observations as
\begin{align*}
  \Sigma^{-1/2} \dataX_i
    &= \Sigma^{-1/2} \mu(G_i) + \Sigma^{-1/2} \varepsilon_i \\
    &= \tilde \mu(G_i) + \tilde \varepsilon_i,
\end{align*}
where $\tilde \mu = \Sigma^{-1/2} \mu$ and
$\tilde \varepsilon_i = \Sigma^{-1/2} \varepsilon_i$. The transformed data
has the same number of clusters, but has noise covariance
$\cov(\tilde \varepsilon_i) = I$. The matrix product $\Gamma \Lambda^{-1/2}$
used in step~\ref{step:transform} is an estimate of $\Sigma^{-1/2}$.

The transformation used in step~\ref{step:transform} uses a random orthogonal
matrix~$Q$, which gets applied to the rows of $\dataX$ after 
multiplying by the estimate of $\Sigma^{-1/2}$. We use this random orthogonal
matrix for two reasons. First, it ensures that in expectation, each transformed cluster mean
$Q \Lambda^{-1/2} \Gamma^\T \mu(g)$ for $g = 1, \dotsc, k$ is uniformly
spread across all $P$ features. This ensures that the self-consistency
conditions on the cluster centers enumerated in Proposition~\ref{prop:self-consistency}
are likely to hold. The second reason for multiplying by $Q$ is to spread any
remaining correlation in the noise evenly (in expectation) across all
dimensions. The latter effect follows since if $Z$ is a random vector with
covariance matrix $\Theta$, then $Q Z$ has covariance matrix $Q \Theta Q^\T$,
which has expectation $E(Q \Theta Q^\T) = \tr(\Theta) I$.

Our correlation correction is not backed by a rigorous theoretical
justification. However, the simulations and empirical validation in
Sections~\ref{sec:simulations} and~\ref{sec:empirical-validation}
demonstrate the effectiveness of our ad-hoc adjustment procedure.

\section{Performance in simulations}
\label{sec:simulations}

\subsection{Overview}

In this section, we perform a set of simulations to evaluate the performance of our
proposed method and the associated correlation correction described in
Section~\ref{sec:corr-correct}.
We compare our method
with a basket of competing methods including the Gap statistic
\citep{tibshirani2001estimating}, Gaussian mixture model-based clustering
\citep{fraley2002model}, the CH-index \citep{calinski1974dendrite}, Hartigan's
statistic \citep{hartigan1975clustering}, the Jump method
\citep{sugar2003finding}, Prediction strength \citep{tibshirani2005cluster},
and Bootstrap stability \citep{fang2012selection}.
We use the default parameter settings for all competing methods. For Gabriel
cross-validation, we perform $2$-fold cross-validation on the columns~($p=q$)
and $5$-fold cross-validation on the rows~($m=n/4$).
We also compare with Wold cross-validation, which we describe in
Appendix~\ref{sec:wold-cv}.

In all simulation settings, we randomly generate cluster centers by drawing
from a multivariate normal distribution with covariance matrix $\tau I$,
conditional on the cluster centers being well-separated (if the distance
between any two cluster centers is less than $1$, then we re-draw a new set of
cluster centers). We choose $\tau$ to make the probability the cluster
centers being well-separated on the first draw to be equal to approximately
50\%. Many of our simulation settings are chosen to mimic the settings used by
\cite{tibshirani2001estimating}.

For each setting, we perform $100$ replicates. We report the number of times
that each method finds the correct number~$k$ of clusters. We also report 95\%
confidence intervals for the proportions, using Wilson's
method~\citep{wilson1927probable}.  The simulations demonstrate that overall,
the proposed Gabriel cross-validation method and its correlation-corrected
version compare well with the competing methods, and they are robust to
variance heterogeneity, high dimensional data, and heavy-tail data.

\subsection{Setting 1: Correlation between dimensions}

We generate six clusters in $10$ dimensions.  Each cluster has $100$ or $50$
multivariate normal observations with common covariance matrix $\Sigma$ which
has compound symmetric structure with $1$ in diagonal and $\rho$ off diagonal.
$\rho$ takes value in $\{0,0.1,...,0.9\}$.
	
\begin{figure}[H]
\centering
\includegraphics[width=5.5in, height=3.3in]{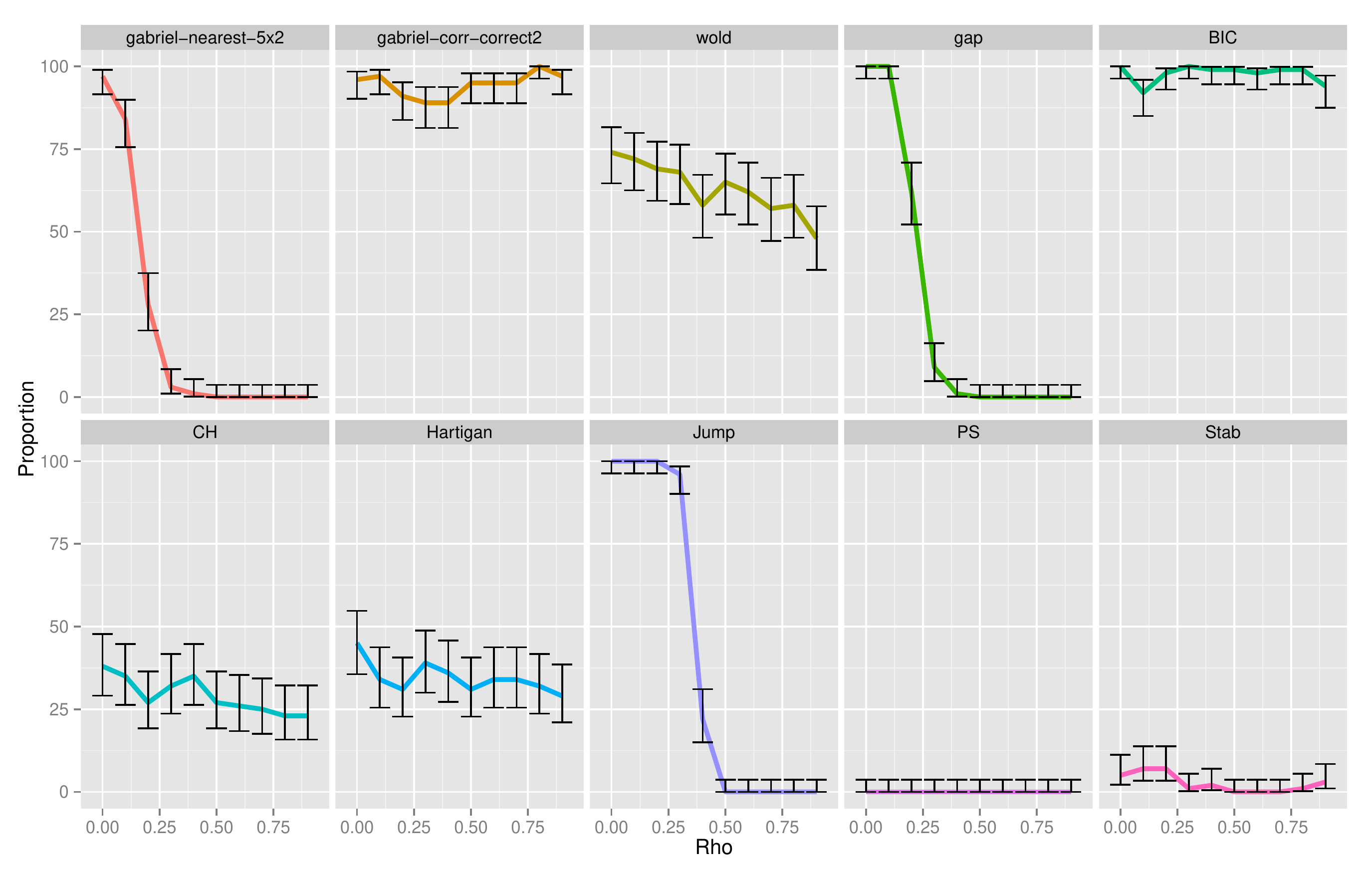}
\label{fig:setting1}
\end{figure}
	
We can see that high correlation between dimensions causes problem for most
existing methods, including Gabriel cross-validation method without the
correlation correction.  The only two methods that work well in the presence
of high correlation are the Gaussian model-based BIC method
\citep{fraley2002model} and the correlation-corrected Gabriel method.

\subsection{Setting 2: Noise dimensions}

We generate three clusters in $6$ dimensions. Each cluster has $1000$ or $500$
multivariate normal observations with identity covariance matrix.
We add $r$ dimensions of noise to the data, randomly generated from a uniform
distribution on $[0,1]$. The noise dimension $r$ takes values in
$\{0,6,...,54\}$. 
	
\begin{figure}[H]
\centering
\includegraphics[width=5.5in, height=3.3in]{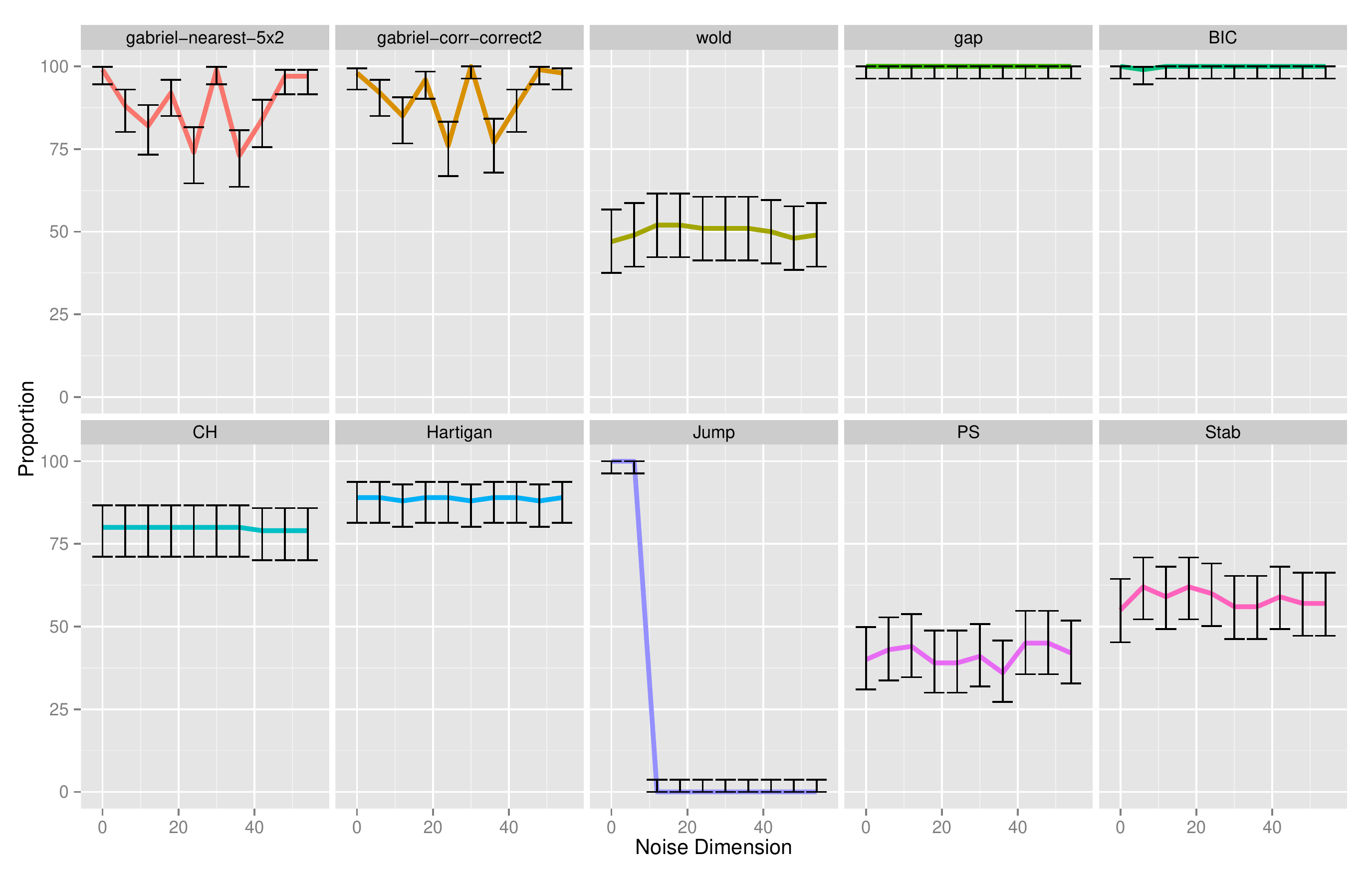}
\label{fig:setting2}
\end{figure}

Most methods are relatively insensitive to adding more noise dimensions; the
one exception to this is the Jump method, which deteriorates significantly in
the presence of extra noise dimensions.

\subsection{Setting 3: High dimension}

We generate eight clusters in $P$ dimensions, with $P$ taking values in
$\{10,20,...,100\}$.  Each cluster has $100$ or $50$ multivariate
normal observations with identity covariance matrix. 
	
\begin{figure}[H]
\centering
\includegraphics[width=5.5in, height=3.3in]{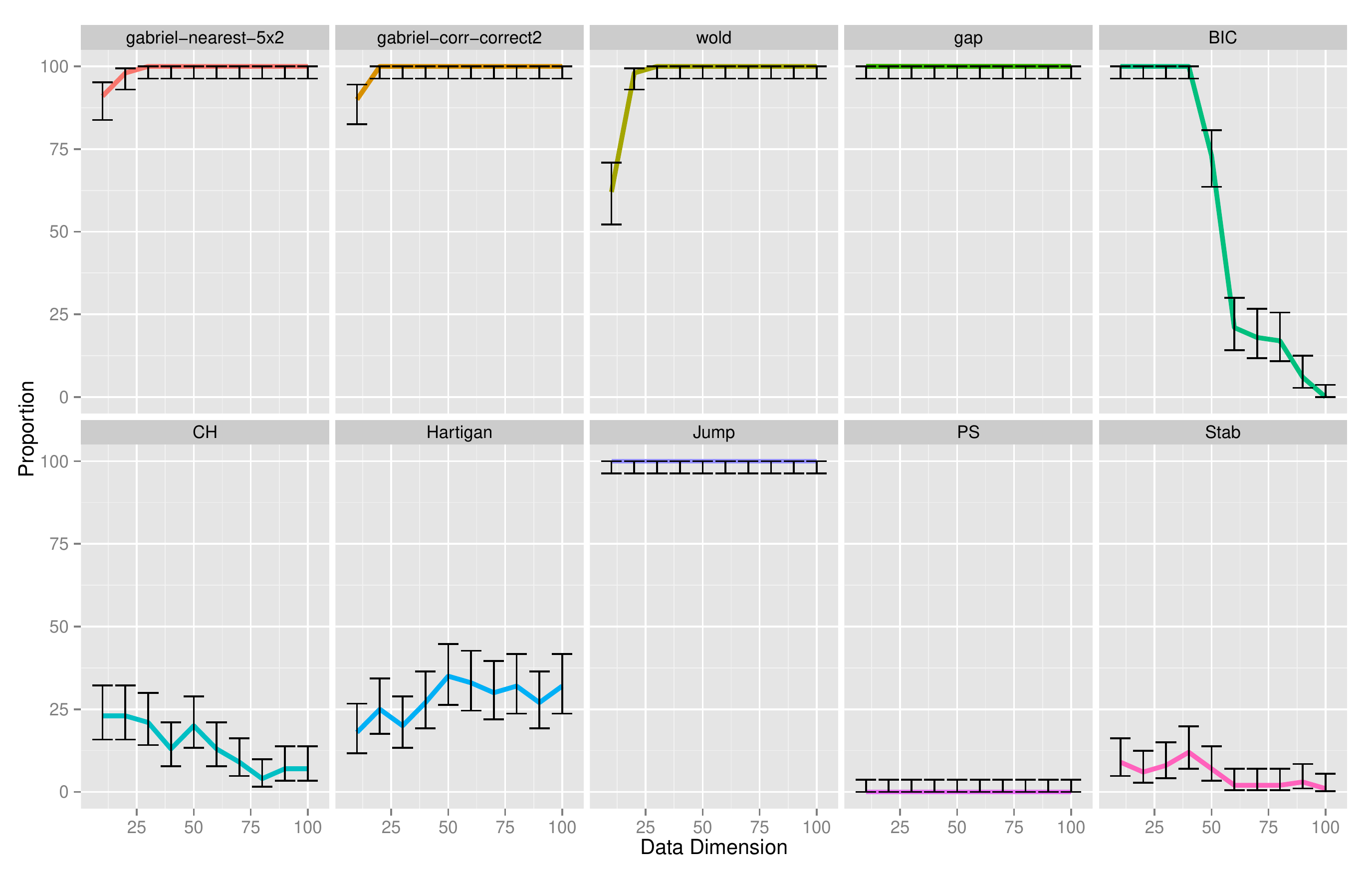}
\label{fig:setting3}
\end{figure}

Some methods, like Jump and Gap, are insensitive to higher dimensions while
other methods deteriorate quickly with increasing dimension, most notably the
Gaussian model-based BIC method.  Gabriel cross-validation and its
correlation-corrected version tend to work better in higher dimensions.

\subsection{Setting 4: Variance heterogeneity}

We generate three clusters in $20$ dimensions. Each cluster has $60$
observations.  Observations are generated from
$\mathcal{N}\left(\mathbf{0},\sigma_1^2\mathbf{I}\right)$,
$\mathcal{N}\left(\mathbf{0},\sigma_2^2\mathbf{I}\right)$ and
$\mathcal{N}\left(\mathbf{0},\sigma_3^2\mathbf{I}\right)$ where $\sigma_1^2 :
\sigma_2^2: \sigma_3^2 = 1:\frac{1+R}{2}:R$. The maximum ratio $R$ takes
values in $\{1,5,10,...,45\}$
	
\begin{figure}[H]
\centering
\includegraphics[width=5.5in, height=3.3in]{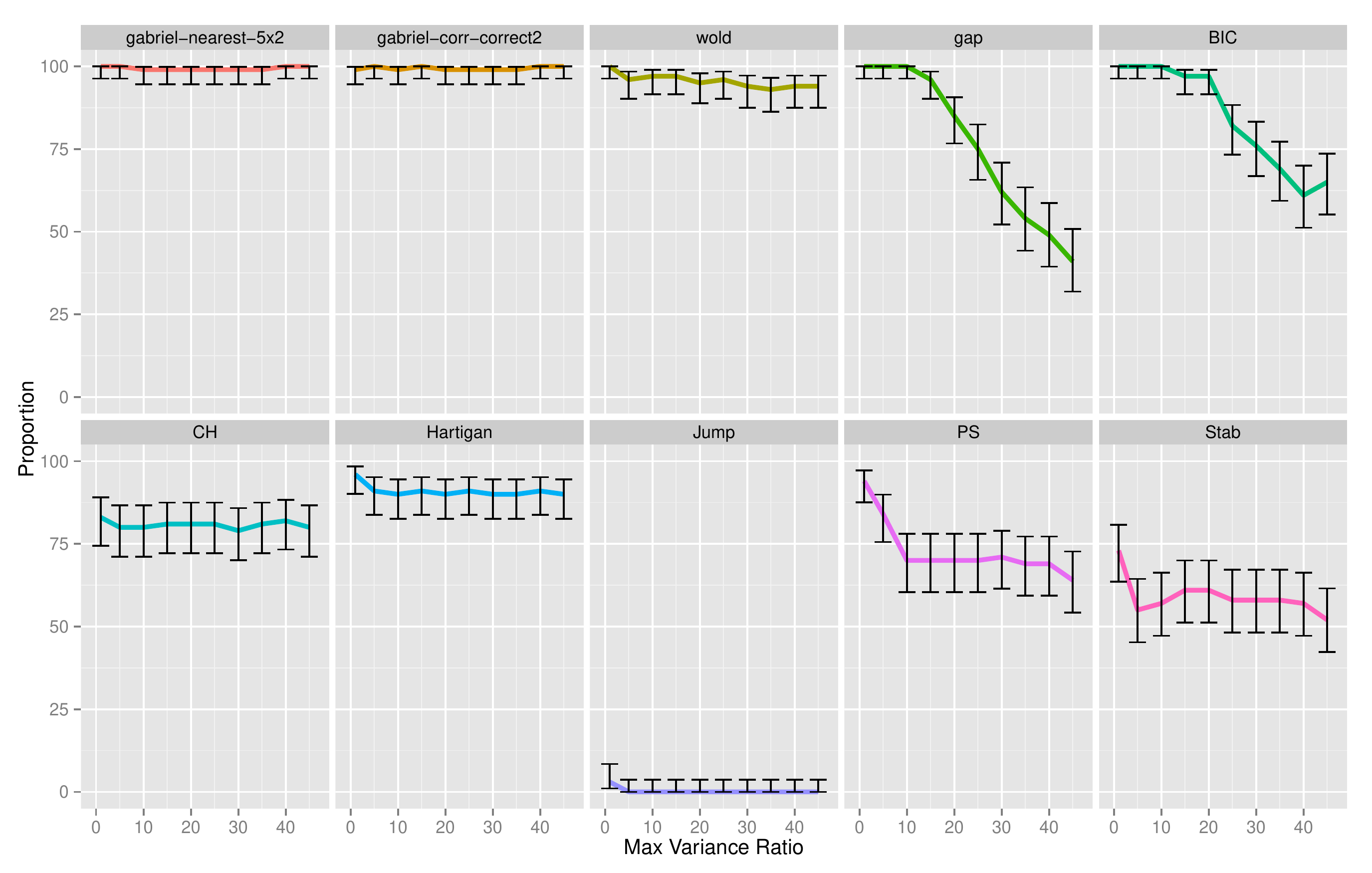}
\label{fig:setting4}
\end{figure}

This setting demonstrates that most existing methods are sensitive to variance
heterogeneity, most notably the Gap method and the model-based BIC method. The
proposed Gabriel cross-validation method and its correlation-corrected version
consistently perform well in estimating $k$ and they are insensitive to
variance heterogeneity.

\subsection{Setting 5: Heavy tail data}

We generate five clusters in $15$ dimensions. Each cluster has $80$
observations.  Observations have independent $t$ distributions in each
dimension, with degrees of freedom $\nu$ taking values in $\{11,10,...,2\}$
     
\begin{figure}[H]
\centering
\includegraphics[width=5.5in, height=3.3in]{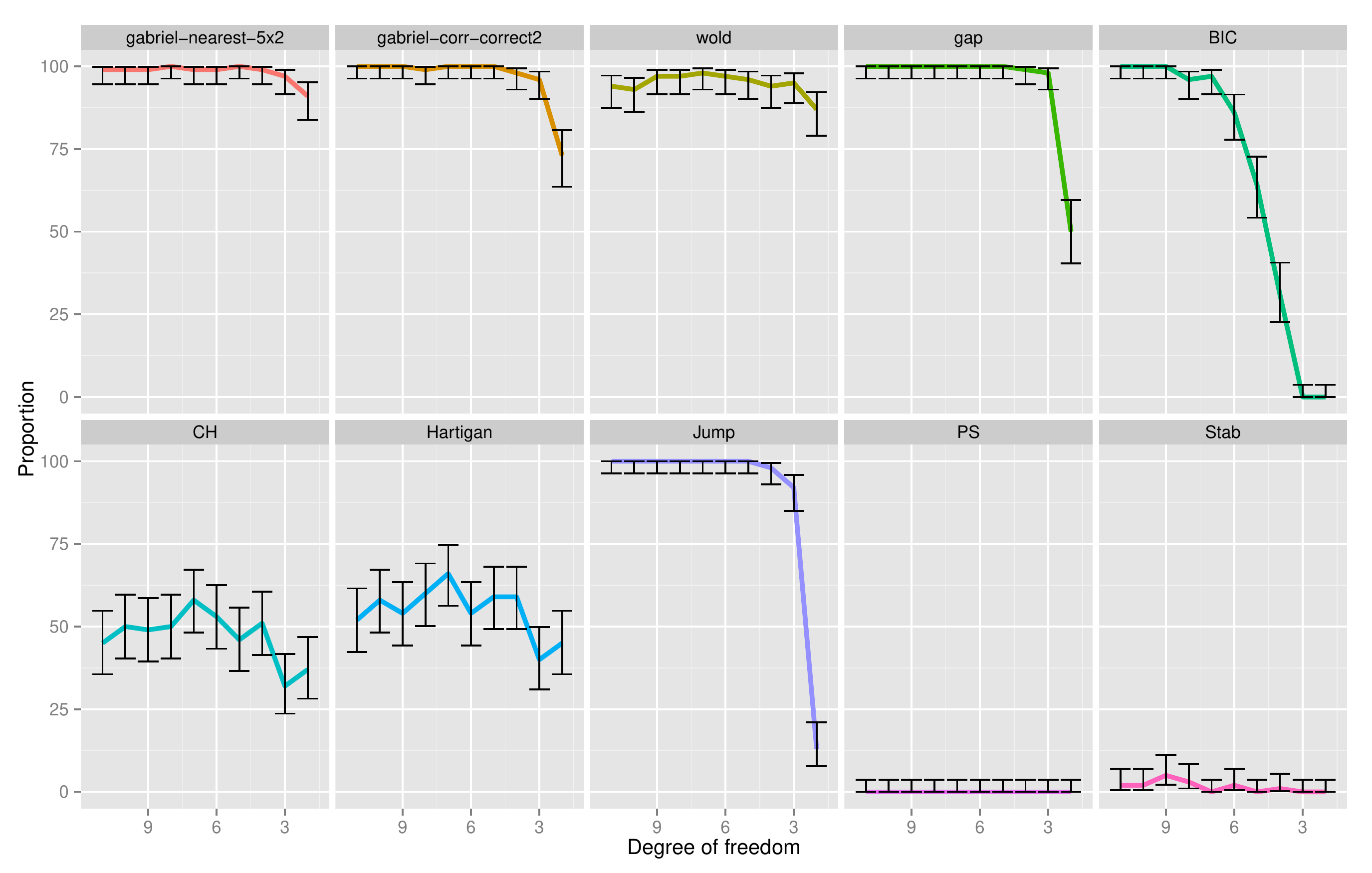}
\label{fig:setting5}
\end{figure}

This setting investigates performance in the presence of heavy-tailed data.
When the degrees of freedom decreases, the tail becomes more flat and the
Gaussian assumption becomes more inappropriate.  For most methods, their
performances are relatively stable until the tail gets very heavy. In the case
where there are $2$ degrees of freedom, the Gap and Jump methods' performances
deteriorate considerably relative to the Gabriel method.

\section{Empirical validation}
\label{sec:empirical-validation}

To further validate our method, we applied it to three real world data sets
with known clustering structure.

The first data set is congressional voting data consisting of voting records
of the second session of the $98$th United States Congress,
\citep{schlimmer1987concept}.  This data set includes votes for legislators on
the $P = 16$ key votes.  For each vote, each legislator either votes
positively (``yea'') or negatively (``nay''). We removed legislators with missing
votes.  This results in $N = 232$ remaining records, with $124$ democrat and
$108$ republican. There are $k = 2$ clusters of legislators, corresponding to
political party.

The second benchmark is the \citet{mangasarian1990pattern} Wisconsin breast
cancer data set.  After excluding the records with missing data, this data set
consists records of $N = 683$ patients, each with measurements of $P = 9$
attributes of their biopsy specimens. It is known that there are $k = 2$
groups of patients: $444$ patients with benign specimens and $239$ patients
with malignant specimens. There is some disagreement as to what the ``true''
value of $k$ should be for this data set; \citet{fujita2014non} have argued
that the malign group is heterogeneous, and should be split into two smaller
clusters, yielding $k = 3$.

The third data set is gene expression data of $k = 5$ types of brain tumors
from \citet{pomeroy2002prediction}, which contains $N = 42$ observations
including $10$ medulloblastomas, $10$ malignant gliomas, $10$ atypical
teratoid/rhabdoid tumors, $8$ primitive neuroectodermal tumours and $4$ normal
cerebella. After preprocessing and feature selection, there are $P = 1379$
variables, corresponding to log activation levels for 1379 genes.


\begin{table}
\centering
\captionsetup{justification=centering}
\caption{\label{tab:benchmark} Number of clusters selected on benchmark datasets}
\begin{tabular}{lccc}
\toprule                  
         & \multicolumn{3}{c}{Dataset} \\
\cmidrule(l){2-4}
Method   & Congress Voting & Breast Cancer & Brain Tumours \\
\midrule                    
Gabriel  & $2$             & $3$           & $5$  \\    
Gabriel (corr.~correct) & $2$ & $2$           & $5$  \\  
Wold     & $2$             & $3$           & $4$  \\
Gap      & $8$             & $10$          & $10$ \\   
BIC      & $2$         & $5$           & $2$  \\   
CH       & $2$             & $2$           & $2$  \\
Hartigan & $3$             & $3$           & $4$  \\
Jump     & $10$            & $9$           & $1$  \\   
PS       & $2$  & $2$           & $1$  \\
Stab.    & $2$  & $2$           & $7$  \\ 
\midrule 
Ground Truth & $2$         & $2$ or $3$    & $5$  \\
\bottomrule
\end{tabular}
\end{table}

We applied the Gabriel cross-validation method, the correlation-corrected
version, and the competing methods described in Section~\ref{sec:simulations}
to each of the three benchmark datasets.  In each dataset, we allowed the
number of clusters, $k$, to range from $1$ to $10$.  Table~\ref{tab:benchmark}
displays the results.  Both versions of the Gabriel method perform well on all
three benchmark datasets. In fact, Gabriel cross-validation is the only method
that correctly identifies the number of clusters in all three benchmark
datasets.

\section{Application to yeast cell cycle data}
\label{sec:application-yeast}

\subsection{Motivation}

Now that we have established that Gabriel cross-validation can effectively
estimate the number of clusters, we apply our method to a yeast cell cycle
dataset. This dataset was collected by \cite{cho1998genome} to study the cell
cycle of budding yeast Saccharomyces cerevisiae. Other authors, including
\citet{tavazoie1999systematic} and \citet{dortet2008model} have used $k$-means
and related methods to cluster the genes in the dataset, with $k$
approximately equal to~30. In both of these analyses, the authors discard the
majority of their clusters as uninterpretable or noise, focusing instead on a
small number of clusters. In contrast to these previous analyses, Gabriel
cross-validation finds a small number, $k = 5$ clusters, all of which are
interpretable.

\subsection{Data collection and preprocessing}

To obtain the raw data, \citet{cho1998genome} first synchronized a collection
of CDC28 yeast cells by raising their temperature to $37^\circ$C in the late
G1 cell cycle phase, then they reinitiated the cell cycle by switching them to
a cooler environment ($25^\circ$C). The authors collected data at $17$ time
points spaced evenly at $10$-minute intervals, covering almost $2$ complete
cell cycles. At each of the $17$ time points, they used
oligonucleotide microarrays to measure $6220$ gene expression profiles.

\citet{tavazoie1999systematic} preprocessed the raw data in an attempt to
normalize the gene responses and remove noise. They reduced the original
$6220$ gene expression profiles to just the $2945$ genes with the highest
variances. Then, they removed the time points at $90$ and $100$ minutes,
because they deemed the measurements at these time points to be unreliable.
Finally, they centered and scaled the genes by subtracting the means and
dividing by the standard deviations, as computed from the remaining $15$ time
points.  After the preprocessing, the data matrix $\dataX$ has $N = 2945$
genes and $P = 15$ time points.

We obtained the preprocessed data and the \citet{tavazoie1999systematic}
cluster analysis from \url{http://arep.med.harvard.edu/network_discovery/}.

\subsection{Clustering}

\begin{figure}
		\centering
	\includegraphics[width=5.7in, height=5.7in]{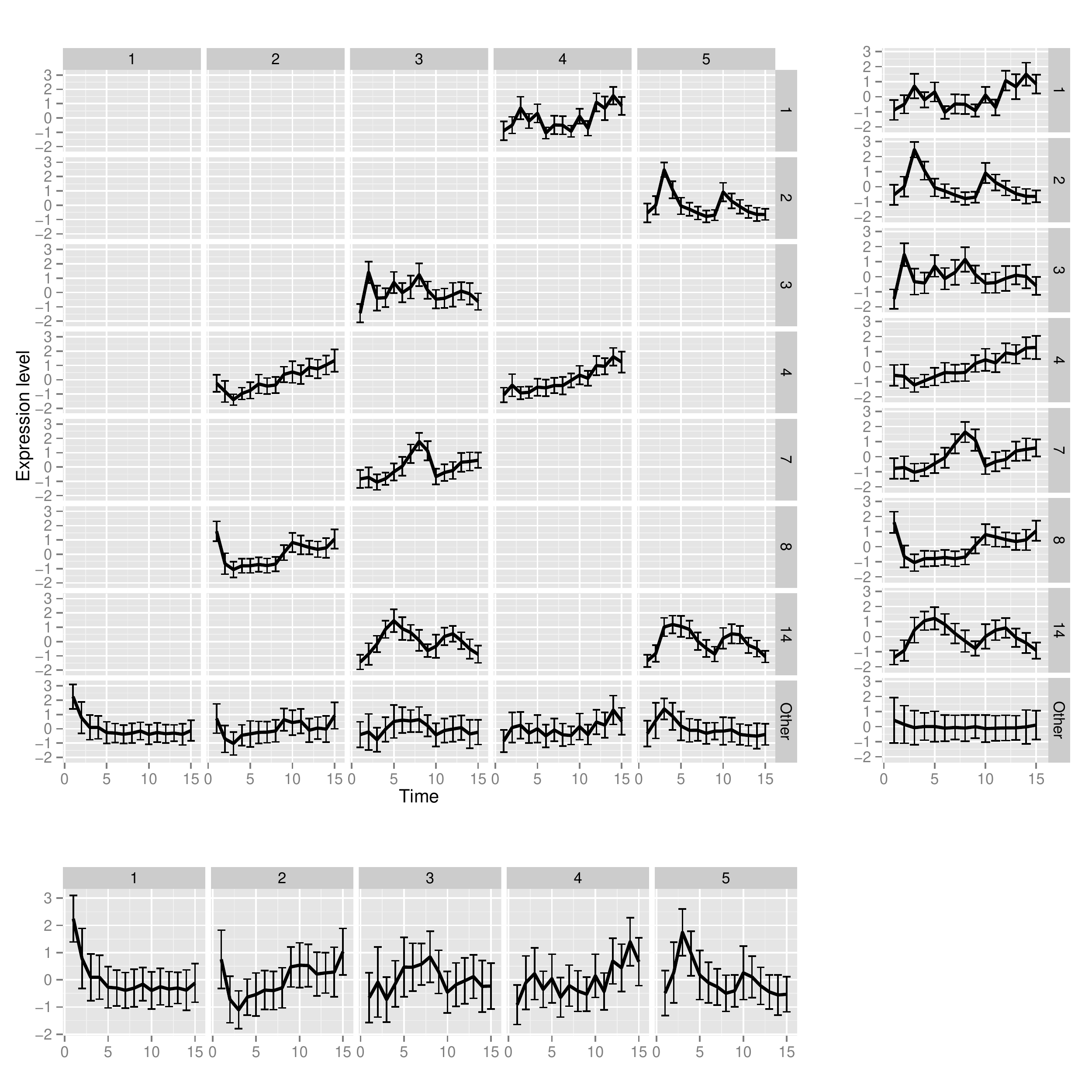}
	\caption{Yeast data set mean expression profiles. The $5$ clusters found by
  the Gabriel method are on the bottom; Clusters profiled in
  \cite{tavazoie1999systematic} are on the right.}
	\label{fig:5_clusters}
\end{figure}	

Following \citet{tavazoie1999systematic} and \citet{dortet2008model}, we
treat the $N = 2945$ gene expression profiles as draws from a mixture
distribution, and we perform $k$-means clustering to segment the genes
according to their expression profiles across the $P = 15$ time points. 
Both the original and the correlation-corrected version of Gabriel
cross-validation find $k = 5$ clusters.

The lower-left panel of Figure~\ref{fig:5_clusters} shows the average
expression level for each cluster across the 15 time points, with error bars
showing standard deviations. Cluster~$1$ has decreasing expression level with
time. The mean gene expression level in Cluster~$2$ decreases at the beginning
and then increases. Cluster~$3$ is a periodic cluster where one can see two
periods corresponding to the two cell cycles. Cluster~$4$ has increasing
expression level with time.  Cluster~$5$ is another periodic cluster.

\subsection{Enrichment analysis}

To further validate our clusters, we follow \citet{tavazoie1999systematic},
performing an enrichment analysis to discover which functional gene groups are
significantly over-represented in each cluster. In the Saccharomyces Genome
Database, each gene is mapped to a set of Gene Ontology categories. We focus
on the $103$ biological process categories.

\begin{table}
\begin{center}
\captionsetup{justification=centering}
\caption{\label{table:enrich} Biological process enrichment within gene clusters}
\footnotesize
\begin{tabular}{ccll}
\toprule
Cluster & Cluster Size & Process Category (In Cluster/Total Genes) & $p$-value \\
\midrule
$1$ & $550$ & response to oxidative stress ($24/55$)  & $1.5\times10^{-5}$  \\
    &       & response to chemical ($64/213$)         & $2.2\times10^{-5}$  \\
$2$ & $590$ & mitochondrion organization ($79/159$)   & $1.1\times10^{-16}$ \\
    &       & mitochondrial translation ($28/51$)     & $2.9\times10^{-8}$  \\
    &       & generation of precursor metabolites and energy ($37/80$) & $7.3\times10^{-8}$ \\
$3$ & $654$ & transcription from RNA polymerase II promoter ($75/214$) & $5.5\times10^{-6}$ \\
    &       & mRNA processing ($30/67$)               & $2.7\times10^{-5}$  \\
    &       & mitotic cell cycle ($63/183$)           & $6.2\times10^{-5}$  \\
$4$ & $634$ & cytoplasmic translation ($105/134$)     & $3.3\times10^{-47}$ \\
    &       & ribosomal subunit biogenesis ($73/138$) & $7.7\times10^{-17}$ \\
    &       & rRNA processing ($61/131$)              & $5.7\times10^{-11}$ \\
    &       & ribosome assembly ($21/36$)             & $1.5\times10^{-6}$  \\
$5$ & $517$ & chromosome segregation ($53/106$)       & $6.0\times10^{-15}$ \\
    &       & cellular response to DNA damage stimulus ($71/172$)      & $3.6\times10^{-14}$ \\
    &       & DNA repair ($64/147$)                   & $3.7\times10^{-14}$ \\
    &       & DNA replication ($42/78$)               & $1.8\times10^{-13}$ \\
    &       & mitotic cell cycle ($70/183$)           & $4.8\times10^{-12}$ \\
\bottomrule
\end{tabular}
\end{center}
\hspace{0.5in} \footnotesize {}
\end{table}

For each category and each cluster, we compute a $p$-value for the null
hypothesis that genes from the category are distributed across all clusters
without any bias towards the particular cluster in question. Under the null
hypothesis, the number of genes from the category that end up in the cluster
is distributed as a hypergeometric random variable. For each cluster, we
compute $p$-values for all $103$ biological process categories, and we report
those that are significantly enrigched in Table~\ref{table:enrich}.  Using a
Bonferroni correction to control the family-wise error rate at level 5\%, we
only report $p$-values that are less than $0.05 / 103 = 4.8\times10^{-4}$.

From Table \ref{table:enrich}, we can see that Cluster~1 is enriched with
genes that somatize cell stress, such as oxidative heat-induce proteins.
Cluster~2 contains genes that govern mitochondrial translation and
mitochondrion organization. Cluster~3, the first period cluster, contains cell
cycle genes related to budding and cell polarity, along with genes that govern
RNA processing and transcription. Cluster~4 contains genes related to
cytoplasmic translation and genes encoding ribosomes. Cluster~5, the second
periodic cluster, contains genes that participate cell-cycle processes, along
with DNA replication and DNA repair.

\subsection{Comparison with Tavazoie~clusters}

In the \citet{tavazoie1999systematic} analysis, those authors performed
$k$-means clustering with $k = 30$; they found 23 of the clusters to be
uninterpretable, and they found 7 clusters to be meaningful. To compare our
clusters with the Tavazoie et al.\ clusters, we prepared a confusion matrix
comparing our clusters with the 7 interpretable Tavazoie clusters in
Table~\ref{table:confusion}.
Entry $(i,j)$ of the confusion matrix gives the number of genes in
Tavazoie's Cluster~$i$ and our Cluster~$j$.

\begin{table}
\caption{\label{table:confusion} Confusion matrix comparing
    the 5 clusters found by Gabriel cross-validation to the 7 interpretable
    clusters found by the \cite{tavazoie1999systematic} analysis}
\centering
\begin{tabular}{lcccccr}
  \toprule
 & Cluster 1 & Cluster 2 & Cluster 3 & Cluster 4 & Cluster 5 & Total \\ 
  \midrule
Cluster 1 & 0 & 0 & 1 & 161 & 2 & 164 \\ 
  Cluster 2 & 1 & 0 & 0 & 0 & 185 & 186 \\ 
  Cluster 3 & 0 & 0 & 91 & 11 & 2 & 104 \\ 
  Cluster 4 & 0 & 102 & 2 & 66 & 0 & 170 \\ 
  Cluster 7 & 1 & 10 & 83 & 7 & 0 & 101 \\ 
  Cluster 8 & 3 & 145 & 0 & 0 & 0 & 148 \\ 
  Cluster 14 & 0 & 1 & 29 & 6 & 38 & 74 \\ 
  Other & 545 & 332 & 448 & 383 & 290 & 1998 \\ 
  \midrule
  Total & 550 & 590 & 654 & 634 & 517 & 2945 \\ 
   \bottomrule
\end{tabular}
\end{table}

Figure~\ref{fig:5_clusters} provides a more in-depth comparison with the
Tavazoie clusters, using a graphical confusion matrix.  The plot in cell
$(i,j)$ of the upper left part of this figure gives the mean expression level
for genes in the intersection of Tavazoie's Cluster~$i$ and our Cluster~$j$;
the plots in the margins give the mean expression levels for Tavazoie's
clusters (top right) and our clusters (bottom left). In
Figure~\ref{fig:5_clusters}, we only include a plot for cell $(i,j)$ if the
number of genes in that cell is greater than $20$.

Our Cluster~1 mainly consists of genes that Tavazoie et al.\ found to be in
uninterpretable clusters. Our Cluster~2 contains high concentrations of
Tavazoie's Clusters~4 and~8. Our first periodic cluster, Cluster~3, contains
high concentrations of Tavazoie's Clulsters~3,~7, and~14; this is notable,
because Tavazoie et al.\ highlighted their Clusters~7 and~14 as being
periodic. Our Cluster~4 contains almost all of Tavazoie's Cluster~1, along
with part of Tavazoie's Cluster~4. Finally, our second periodic cluster,
Cluster~5, contains almost all of Tavazoie's Cluster~2, along with part of
Tavazoie's Cluster 14; this, again, is notable, because Tavazoie et al.\
highlited these clusters as being periodic.

For the clusters that Tavazoie et al.\ were able to characterize, our analysis
broadly agrees with the earlier clustering. The major difference between our
analysis and that of \citet{tavazoie1999systematic} is that we are able to
identify meaningful groups of genes with a much smaller value of $k$ ($k = 5$
instead of $k = 30$), and we are able to interpret all of the clusters
found by our analysis.

\section{Discussion}
\label{sec:discussion}

In this paper, we proposed a new approach to estimate the number of clusters
to be used in $k$-means clustering. The intuition behind our proposed method
is to transform the unsupervised learning problem into a supervised learning
problem via a form of Gabriel cross validation.  We proved that our method is
self-consistent, and we analyzed its behavior in some special cases of
Gaussian mixture models.  Using simulations and real data examples, we
demonstrated that our method has good performance, competitive with existing
approaches. The simulations and empirical benchmarks demonstrate the
advantages of our method. In the yeast cell cycle application, our method
was able to identify meaningful gene groups with a small number of clusters.

There are many other clustering algorithms that get used in practice besides
$k$-means. We suspect that it should be possible to apply our method in the
context of a spectral clustering procedure, after transforming by the
eigenvectors of the Laplacian matrix. For other clustering schemes, including
versions of hierarchical clustering, we are less certain about the viability
of Gabriel cross-validation.  It is an open question as to whether Gabriel
cross-validation can be extended to other clustering methods, and whether such
extensions will perform well in practice.

For $k$-means clustering, Gabriel cross-validation is competitive with
other model selection methods, especially in the presence of high-dimensional,
heterogeneous, or heavy-tailed data.

\section*{Acknowledgements}

We thank Rob Tibshirani for getting us started on this problem and
for providing code for some initial simulations. We thank Art~Owen for
providing us with a summary of the relevant theory on $k$-means clustering,
and for giving us feedback on our theoretical results.  We also thank Cliff
Hurvich, Josh Reed, and Jeff~Simonoff, for providing comments on an early
draft of this manuscript and for suggesting further avenues of inquiry.

\appendix

\section{Clustering scree plot examples}

\label{sec:elbow-fail}

The top row of Figure~\ref{fig:elbow}
displays an example where the elbow in $W_k$ corresponds to the true number $k
= 4$ of mixture components in the data-generating mechanism. The elbow
approach is simple and often performs well, but it requires subjective
judgment as to where the elbow is located, and, as the bottom row of
Figure~\ref{fig:elbow} demonstrates, the approach can easily fail.

\begin{figure}
\centering
\begin{minipage}{\linewidth}
  \begin{minipage}{0.45\linewidth}
    \includegraphics[width=\linewidth]{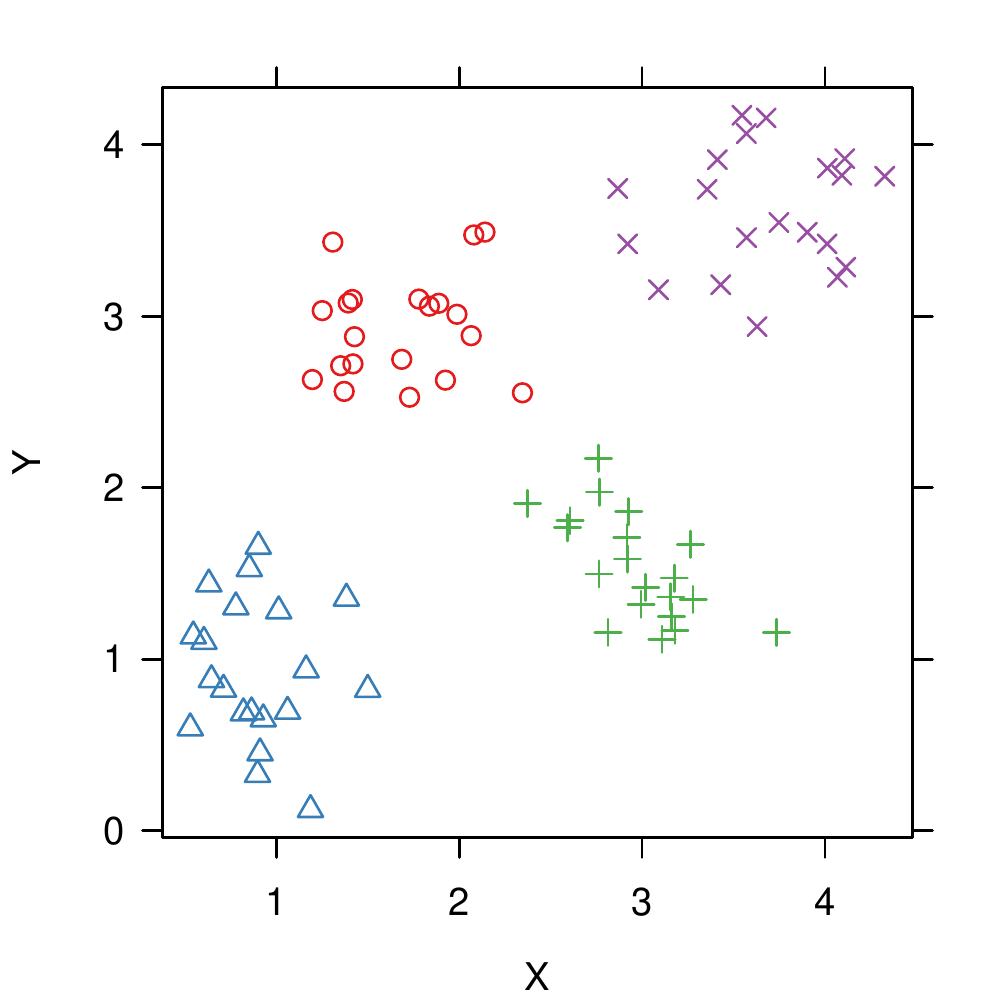}
  \end{minipage}
  \hspace{0.05in}
  \begin{minipage}{0.45\linewidth}
    \includegraphics[width=\linewidth]{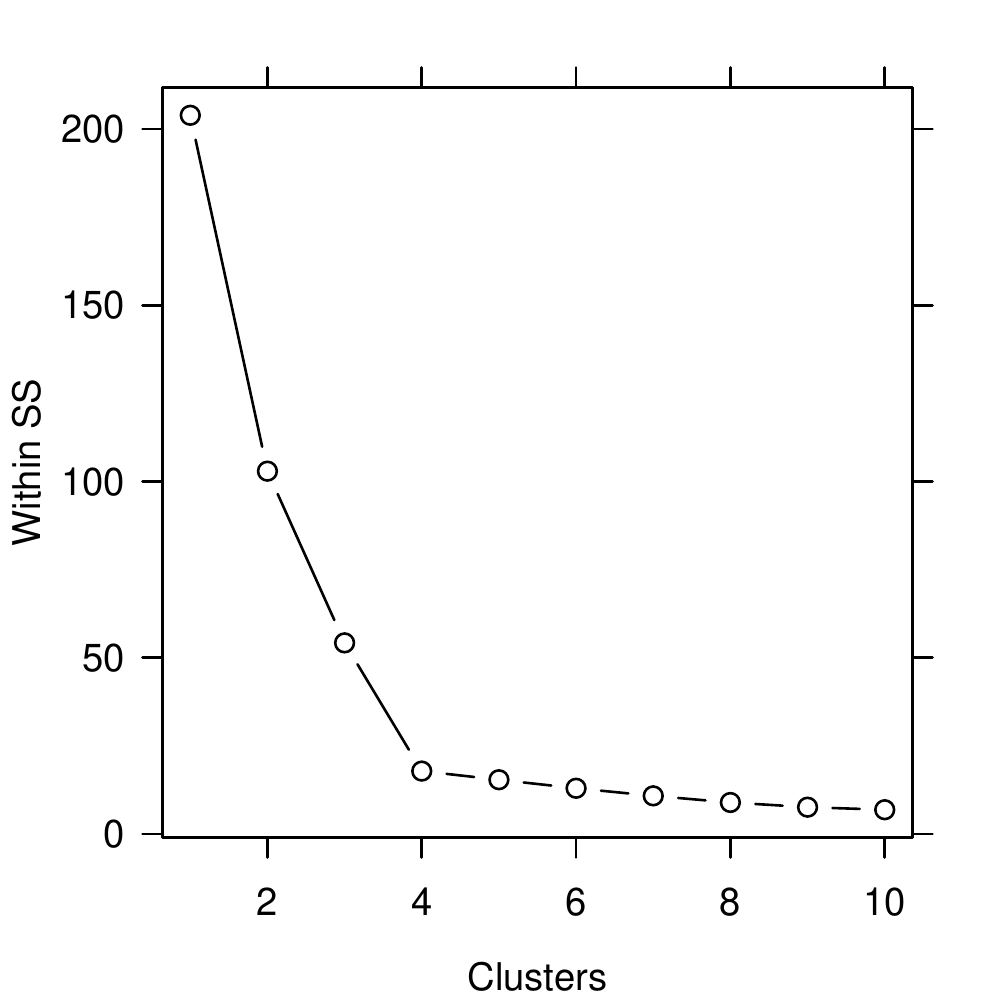}
  \end{minipage}
\end{minipage}
\begin{minipage}{\linewidth}
  \begin{minipage}{0.45\linewidth}
    \includegraphics[width=\linewidth]{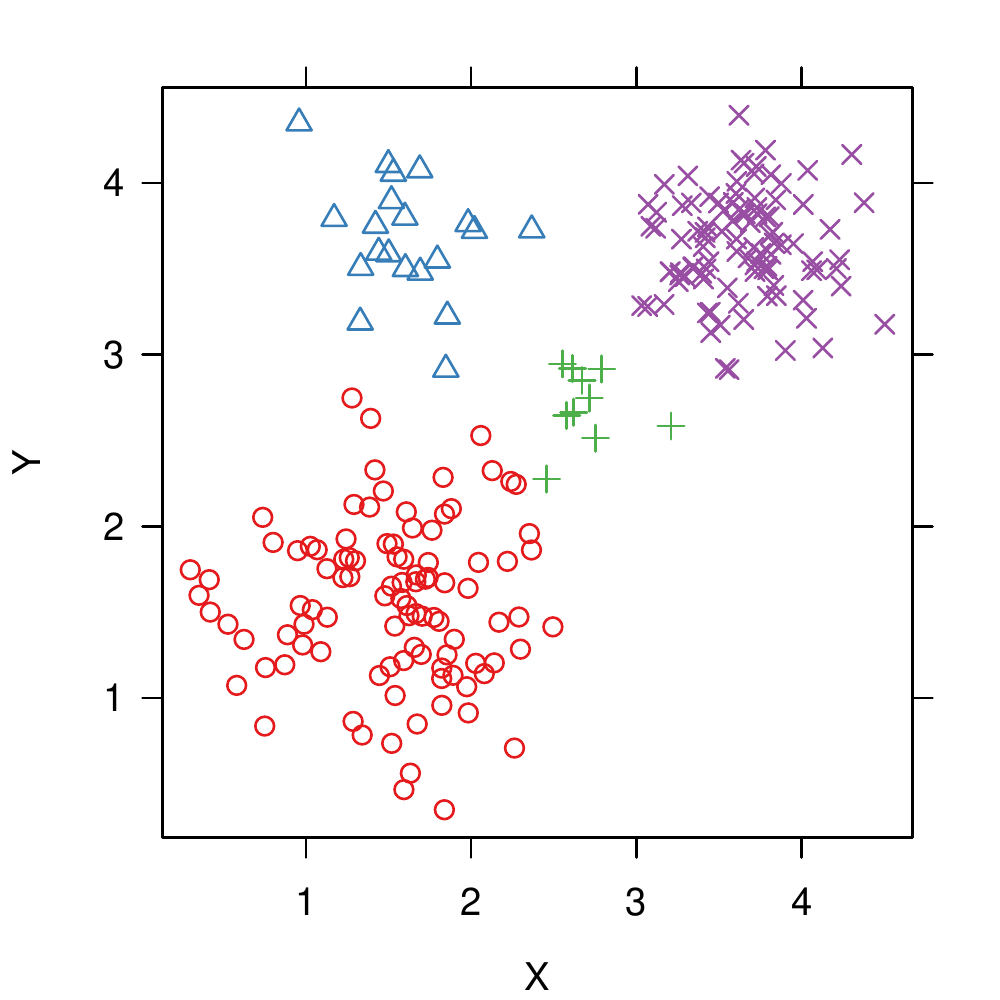}
  \end{minipage}
  \hspace{0.05in}
  \begin{minipage}{0.45\linewidth}
    \includegraphics[width=\linewidth]{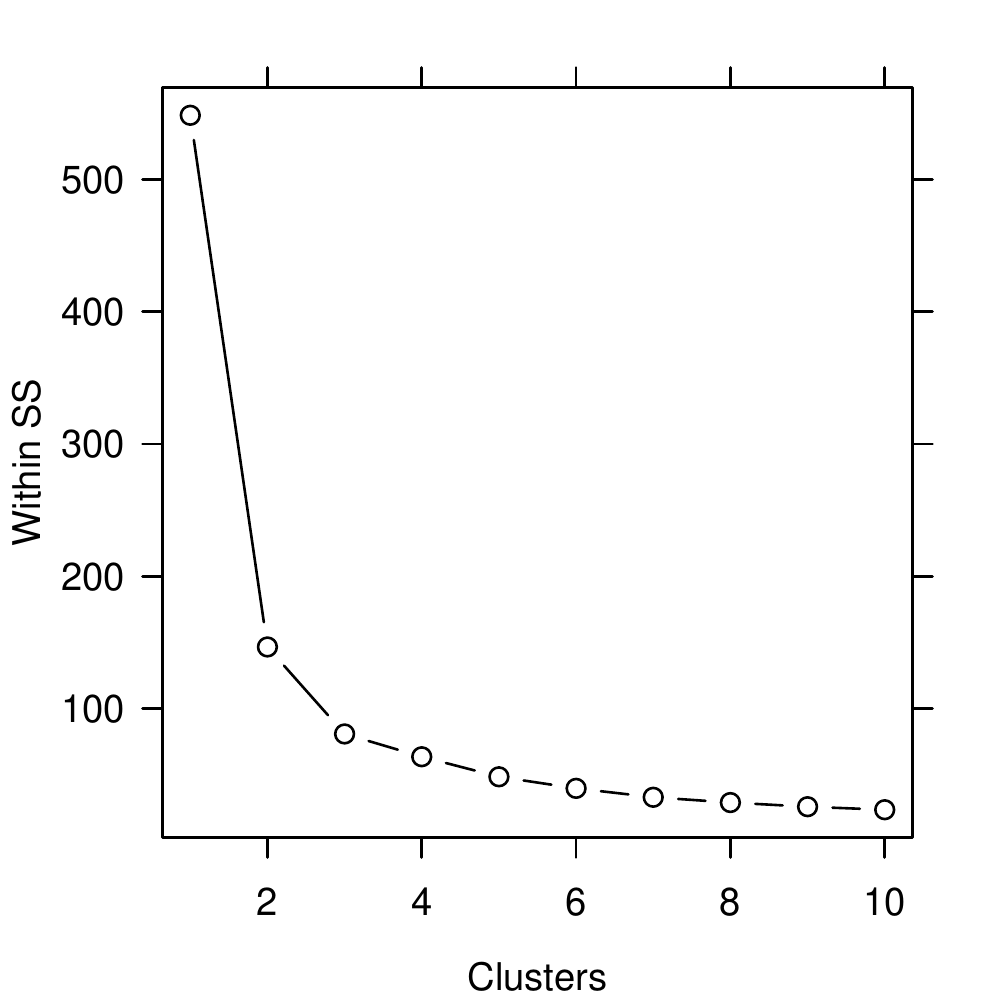}
  \end{minipage}
\end{minipage}
\caption{Left panels show the $(X,Y)$ data points; right panels
  show the corresponding values of the within-cluster sum of squares $W_k$
plotted against the number of clusters, $k$.}
\label{fig:elbow}
\end{figure}

\section{Analysis of Gabriel method: Single cluster in more than two dimensions}
\label{sec:single-general-dim}

\begin{proposition}
Suppose that $\{ (X_i, Y_i) \}_{i=1}^{n + m}$ is data from a single fold
of Gabriel cross-validation, where each $(X,Y)$ pair in $\R^{p+q}$ is an
independent draw from a mean-zero multivariate normal distribution with 
covariance matrix $\Sigma_{XY}= \left( \begin{smallmatrix} \Sigma_{XX} & \Sigma_{XY} \\ 
 \Sigma_{YX} & \Sigma_{YY} \end{smallmatrix}\right)$, with $\Sigma_{YY}$ has leading 
eigenvalue $\lambda_1$ and corresponding eigenvector $u_1$. In this case, the data are drawn
from a single cluster; the true number of clusters is~$1$.  If  $\frac{\sqrt{\lambda_1}}{2}
 > \frac{u^T_1\Sigma_{YX}\Sigma_{XY}u_1}{\sqrt{u^T_1\Sigma_{YX} \Sigma_{XX} \Sigma_{XY} u_1}}$,
then $\CV(1) < \CV(2)$ with probability tending to one as $m$ and $n$ increase.
\end{proposition}

\begin{proof}
Let $X$ and $Y$ be jointly multivariate normal distributed with mean $\mathbf{0}$ and covariance matrix $\Sigma_{XY}$, i.e.
\[	(X,Y) \sim \mathcal{N} \left( \mathbf{0}, \Sigma_{XY}\right)	\]
where $\Sigma_{XY}=\begin{bmatrix} \Sigma_{XX} & \Sigma_{XY} \\  \Sigma_{YX} & \Sigma_{YY} \end{bmatrix}$.

Let $\Sigma_{YY} = U \Lambda U^T$  be the eigendecomposition of $\Sigma_{YY}$,
with leading eigenvalue $\lambda_1$ and corresponding eigenvector $u_1$. Then
the centroid of $k$-means applying on $(y_1,..,y_n)$ is on the first principal
component
of $Y$,\[	E(u^T_1 Y|u^T_1 Y>0) = \bar{\mu}^Y_1 =\sqrt{2 \lambda_1/\pi}u_1\] and 
\[	E(u^T_1 Y|u^T_1 Y<0) = \bar{\mu}^Y_2 =-\sqrt{2 \lambda_1/\pi}u_1\]
 where $u^T_1 Y \sim \mathcal{N}(0,\lambda_1)$.

To compute $\bar{\mu}^X_1 = E(X|u^T_1 Y>0)$, we need to know the conditional distribution $X|u^T_1 Y$. Since $(X,Y)$ has multivariate normal distribution, $(X,u^T_1 Y)$ also has a multivariate normal distribution with mean $\mathbf{0}$ and covariance matrix
$$\Sigma_{X,u^T_1 Y}=\begin{bmatrix} \Sigma_{XX} & \Sigma_{XY} u_1 \\  u^T_1 \Sigma_{YX} & \lambda_1 \end{bmatrix}$$
The conditional distribution $X|u^T_1 Y$ is hence normal with mean
 $$\mu_{X|u^T_1 Y} = \Sigma_{XY} u_1 \lambda^{-1}_1 u^T_1 Y $$
Therefore, 
\begin{align}
\bar{\mu}^X_1 &= E(X \mid u^T_1 Y>0) \nonumber \\ \nonumber
 			  &= E\left(E[X \mid u^T_1 Y] \mid u^T_1 Y>0\right) \\ \nonumber
 			  &=  E\left(\Sigma_{XY} u_1 \lambda^{-1}_1 u^T_1 Y \mid u^T_1 Y>0\right)\\ \nonumber
 			  &= \lambda^{-1}_1 \Sigma_{XY}u_1 E(u^T_1 Y \mid u^T_1 Y>0) \\ \nonumber
 			  &= \lambda^{-1}_1 \Sigma_{XY}u_1 \sqrt{2 \lambda_1/\pi} \\ \nonumber
      &= \sqrt{2 / (\lambda_1 \pi)} \Sigma_{XY}u_1
\end{align}
Similar calculation yields $\bar{\mu}^X_2 = -\sqrt{2 / (\lambda_1 \pi)} \Sigma_{XY}u_1$.
The decision rule to classify any observed value of $X$ to $\bar{\mu}^X_1$ is therefore
\[	(\bar{\mu}^X_1)^T X >0	\hspace{0.2in}\text{or} \hspace{0.2in} u^T_1\Sigma_{YX}X>0\] 
Since $u^T_1\Sigma_{YX}X$ is a linear combination of $X$, it also has normal distribution 
\[	\mathcal{N} \left( 0, u^T_1\Sigma_{YX} \Sigma_{XX} \Sigma_{XY} u_1\right)	\]
And $(Y,u^T_1\Sigma_{YX}X)$ also have multivariate normal distribution with mean $\mathbf{0}$ 
and covariance matrix
\[
\begin{bmatrix}
\Sigma_{YY} & \Sigma_{YX}\Sigma_{XY}u_1  \\
u^T_1\Sigma_{YX}\Sigma_{XY} &  u^T_1\Sigma_{YX} \Sigma_{XX} \Sigma_{XY} u_1
\end{bmatrix}
\]
The conditional distribution of $Y|u^T_1\Sigma_{YX}X$ is also multivariate normal with mean 
\[	
\mu_{Y|u^T_1\Sigma_{YX}X } = \Sigma_{YX}\Sigma_{XY}u_1 (u^T_1\Sigma_{YX} \Sigma_{XX} \Sigma_{XY} u_1)^{-1}u^T_1\Sigma_{YX}X	
\]
The $Y$ center for $u^T_1\Sigma_{YX}X>0$ is
\begin{align}
\hat{\mu}^Y_1 &= E(Y|u^T_1\Sigma_{YX}X>0) \nonumber \\ \nonumber
 		      & =  \Sigma_{YX}\Sigma_{XY}u_1 (u^T_1\Sigma_{YX} \Sigma_{XX} \Sigma_{XY} u_1)^{-1} E(u^T_1\Sigma_{YX}X \mid u^T_1\Sigma_{YX}X>0) \\ \nonumber
\end{align}
Note that $u^T_1\Sigma_{YX}X$ has normal distribution $\mathcal{N} \left( 0, u^T_1\Sigma_{YX} \Sigma_{XX} \Sigma_{XY} u_1\right)$, so
\[
E(u^T_1\Sigma_{YX}X \mid u^T_1\Sigma_{YX}X>0) = \sqrt{2/\pi}\cdot\sqrt{u^T_1\Sigma_{YX} \Sigma_{XX} \Sigma_{XY} u_1}
\]
Therefore, we have the $Y$ center for $u^T_1\Sigma_{YX}X>0$ be
\begin{align*}
\hat{\mu}^Y_1 &= \sqrt{2/\pi}\cdot\sqrt{u^T_1\Sigma_{YX} \Sigma_{XX} \Sigma_{XY} u_1} \hspace{0.1in} \Sigma_{YX}\Sigma_{XY}u_1 (u^T_1\Sigma_{YX} \Sigma_{XX} \Sigma_{XY} u_1)^{-1} \\
&=\frac{\sqrt{2/\pi}}{\sqrt{u^T_1\Sigma_{YX} \Sigma_{XX} \Sigma_{XY} u_1}} \Sigma_{YX}\Sigma_{XY}u_1
\end{align*} 
 
Recall that $\bar{\mu}^Y_1 =\sqrt{2 \lambda_1/\pi}u_1$, to judge if $\CV(2) >
\CV(1)$, one only need to compare the distance between  $\hat{\mu}^Y_1$  and
$\bar{\mu}^Y_1$ with distance between  $\hat{\mu}^Y_1$ and grand mean $0$. By
the variance and bias decomposition of prediction MSE, when variance is the
same, only bias influences the MSE. 

After some linear algebra manipulation, we get
$\|\hat{\mu}^Y_1 - \bar{\mu}^Y_1\|^2 > \|\hat{\mu}^Y_1\|^2$ or $\CV(2) >
\CV(1)$ if and only if
\[
 \frac{\sqrt{\lambda_1}}{2} > \frac{u^T_1\Sigma_{YX}\Sigma_{XY}u_1}{\sqrt{u^T_1\Sigma_{YX} \Sigma_{XX} \Sigma_{XY} u_1}} 
\]
\end{proof}

\section{Technical Lemmas}
\label{app:technical-lemmas}

\begin{lemma}\label{lem:truncated-normal-moments}

If $Z$ is a standard normal random variable, then
\[
  \E(Z \mid a < Z < b)
    = - \frac{\varphi(b) - \varphi(a)}
             {\Phi(b) - \Phi(a)}
\]
and
\[
  \E\{(Z - \delta)^2 \mid a < Z < b\}
    = \delta^2 + 1
    - \frac{  (b - 2 \delta) \varphi(b)
            - (a - 2 \delta) \varphi(a)}
           {\Phi(b) - \Phi(a)}
\]
for all constants $a$, $b$, and $\delta$, where $\varphi(z)$ and $\Phi(z)$ are
the standard normal probability density and cumulative distribution functions.
These expressions are valid for $a = -\infty$ or $b = \infty$ by taking
limits.

\end{lemma}
\begin{proof}
We will derive the expression for the second moment.  Integrate to get
\begin{align*}
  \E[ (Z - \delta)^2 1\{Z < b\}]
    &= \int_{-\infty}^b (z - \delta)^2 \varphi(z) \, dz \\
    &= (\delta^2 + 1) \Phi(b) - (b - 2 \delta) \varphi(b).
\end{align*}
Now,
\[
  \E\{(Z - \delta)^2 \mid a < Z < b\}
    =
    \frac{  \E[ (Z - \delta)^2 1\{Z < b\}]
          - \E[ (Z - \delta)^2 1\{Z < a\}]}
         { \Phi(b) - \Phi(a) }.
\]
\end{proof}

Lemma~\ref{lem:truncated-normal-moments} has some important special cases:
\begin{align*}
  \E\{Z \mid Z > 0\} &= 2 \varphi(0) = \sqrt{2 / \pi}, \\
  \E\{(Z - \delta)^2 \mid Z > 0 \}
    &= \delta^2 + 1 - 4 \delta \varphi(0), \\
  \E\{(Z - \delta)^2 \mid Z < 0 \}
    &= \delta^2 + 1 + 4 \delta \varphi(0).
\end{align*}

\section{Wold cross-validation}
\label{sec:wold-cv}

In Wold cross-validation, we perform ``speckled'' hold-outs in each fold, leaving out a
random subset of the entries of the data matrix $\dataX \in \R^{N \times P}$. For each value
of $k$ and each fold, we perform the following set of actions to get an
estimate of cross-validation error, $\CV(k)$, which we average over all folds.

\begin{enumerate}
  \item Randomly partition the set of indices $\{ 1, 2, \dotsc, N \} \times \{
    1, 2, \dotsc, P \}$ into a train set $S_\text{train}$ and a test set
    $S_\text{test}$.

  \item Apply a $k$-means fitting procedure that can handle missing data to
    the training data $\{ \dataX_{i,j} : (i,j) \in S_{\text{train}} \}$. This
    gives a set of cluster means $\mu(1), \dotsc, \mu(k) \in \R^P$ and
    cluster labels for the rows, $G_1, G_2, \dotsc, G_N$.

  \item Compute the cross-validation error as
    \[
      \CV(k) = \sum_{(i,j) \in S_{\text{test}}} \{ \dataX_{i,j} - \mu_j(G_i) \}^2,
    \]
    where $\mu_j(G_i)$ denotes the $j$th component of $\mu(G_i)$.
\end{enumerate}

\bibliographystyle{apalike}

\bibliography{references}
\end{document}